\documentclass[sigconf, nonacm]{acmart}

\newcommand\vldbdoi{XX.XX/XXX.XX}
\newcommand\vldbpages{XXX-XXX}
\newcommand\vldbvolume{20}
\newcommand\vldbissue{1}
\newcommand\vldbyear{2026}
\newcommand\vldbauthors{\authors}
\newcommand\vldbtitle{\shorttitle} 
\newcommand\vldbavailabilityurl{URL_TO_YOUR_ARTIFACTS}
\newcommand\vldbpagestyle{plain} 

\usepackage{fancyhdr}
\usepackage[linesnumbered,vlined,ruled]{algorithm2e}
\usepackage{algpseudocode}
\usepackage{graphicx}
\usepackage{textcomp}
\usepackage{xcolor}
\usepackage{bm}
\usepackage{booktabs}
\usepackage{multirow}
\usepackage{enumitem}
\usepackage{microtype}
\usepackage{subfigure}
\usepackage{booktabs} 
\usepackage{amsfonts} 

\usepackage{amssymb}
\usepackage{amsmath}
\usepackage{amsthm}
\usepackage{sansmath}
\usepackage{subfigure}
\usepackage{amsthm}
\usepackage{tcolorbox}
\usepackage{colortbl}
\usepackage{tikz}
\usepackage{url}
\usepackage{pifont}
\usepackage{makecell}
\usepackage{cleveref}
\usepackage{pgfplots}
\usepackage{setspace}
\usepackage{hyperref}
\usepackage[dvipsnames]{xcolor}
\usepackage{soul}

\newcommand{\CSV}{\textsc{CSV}\xspace}
\newcommand{\CSVU}{\textsc{UniCSV}\xspace}
\newcommand{\CSVS}{\textsc{SimCSV}\xspace}

\newcommand{\basefilter}{{Reference}\xspace}

\newcommand{\lotus}{{Lotus}\xspace}
\newcommand{\BARGAIN}{{BARGAIN}\xspace}

\crefformat{section}{\S#2#1#3} 
\crefformat{subsection}{\S#2#1#3}
\long\def\comment#1{}

\AtBeginDocument{%
  }

\setcopyright{acmcopyright}
\copyrightyear{2026}
\acmYear{2026}
\acmDOI{XXXXXXX.XXXXXXX}

\acmConference[SIGMOD '26]{the 2026 International Conference on Management of Data}{May 31--June 5,
  2026}{Bengaluru, India}

\acmPrice{15.00}
\acmISBN{978-1-4503-XXXX-X/18/06}

\begin{document}

\title{Beyond Linear LLM Invocation: An Efficient and Effective Semantic Filter Paradigm}

\author{Nan Hou}
\affiliation{%
  \institution{The Chinese University of Hong Kong}
  \country{}
}
\email{nhou@se.cuhk.edu.hk}

\author{Kangfei Zhao}
\affiliation{%
  \institution{Beijing Institute of Technology}
  \country{}
}
\email{zkf1105@gmail.com}

\author{Jiadong Xie}
\affiliation{%
  \institution{The Chinese University of Hong Kong}
  \country{}
}
\email{jdxie@se.cuhk.edu.hk}

\author{Jeffrey Xu Yu}
\affiliation{%
  \institution{HKUST (Guangzhou)}
  \country{}
}
\email{jeffreyxuyu@hkust-gz.edu.cn}

\newcommand{\nthesection}{\arabic{section}}

\newcounter{remark}[section]
\renewcommand{\theremark}{\nthesection.\arabic{remark}}
\newenvironment{remark}{\begin{em}
        \refstepcounter{remark}
        {\vspace{1ex}\noindent\bf Remark \theremark:}}{
        \end{em}\vspace{1ex}} 



\newcommand{\proofsketch}{\noindent{\bf Proof Sketch: }}
\newcommand{\myproof}{\noindent{\bf Proof: }}


\newcommand{\eop}{\hspace*{\fill}\mbox{$\Box$}\vspace*{1ex}}
\newcommand{\stitle}[1]{\vspace{1ex} \noindent{\bf #1}}
\newcommand{\etitle}[1]{\vspace{1ex}\noindent{\underline{\em #1}}}
\newcommand{\kw}[1]{{\ensuremath {\mathsf{#1}}}\xspace}

\newcommand{\mat}[2]{{\begin{tabbing}\hspace{#1}\=\+\kill #2\end{tabbing}}}
\newcommand{\m}{\hspace{0.05in}}
\newcommand{\ls}{\hspace{0.1in}}
\newcommand{\beqn}{\begin{eqnarray*}}
\newcommand{\eeqn}{\end{eqnarray*}}

\newcounter{ccc}
\newcommand{\bcc}{\setcounter{ccc}{1}\theccc.}
\newcommand{\icc}{\addtocounter{ccc}{1}\theccc.}

\newcommand{\tabincell}[2]{\begin{tabular}{@{}#1@{}}#2\end{tabular}} 

\newcommand{\kwnospace}[1]{{\ensuremath {\mathsf{#1}}}}

\newcommand{\radd}[1]{\textcolor{red}{#1}}
\newcommand{\revise}[1]{\textcolor{blue}{#1}}
\newcommand{\cmark}{\ding{51}}%
\newcommand{\xmark}{\ding{55}}%

\newcommand{\bigo}{\ensuremath{\mathcal{O}}}%

\newcommand{\AGM}{\kw{AGM}}
\newcommand{\GHD}{\kw{GHD}}
\newcommand{\ALSS}{\kw{ALSS}}
\newcommand{\LSS}{\kw{LSS}}
\newcommand{\AL}{\kw{AL}}
\newcommand{\MLP}{\kw{MLP}}
\newcommand{\IN}{\kw{IN}}
\newcommand{\LSSFRE}{{$\texttt{LSS-fre}$}\xspace}
\newcommand{\LSSEMB}{{$\texttt{LSS-emb}$}\xspace}
\newcommand{\LSSCON}{{$\texttt{LSS-con}$}\xspace}
\newcommand{\EmptyHeaded}{{\sl EmptyHeaded}\xspace}
\newcommand{\Graphflow}{{\sl Graphflow}\xspace}
\newcommand{\GraphQL}{{\sl GraphQL}\xspace}
\newcommand{\wordtovec}{{\sl word2vec}\xspace}
\newcommand{\ProtBert}{{\sl ProtBert}\xspace}
\newcommand{\DeepWalk}{{\sl DeepWalk}\xspace}
\newcommand{\prone}{{\sl ProNE}\xspace}
\newcommand{\GCN}{{\sl GCN}\xspace}
\newcommand{\GNN}{{\sl GNN}\xspace}
\newcommand{\GIN}{{\sl GIN}\xspace}
\newcommand{\GAT}{{\sl GAT}\xspace}
\newcommand{\SAGE}{{\sl GraphSAGE}\xspace}
\newcommand{\GCARE}{{\sl G-CARE}\xspace}
\newcommand{\WJ}{{$\texttt{WJ}$}\xspace}
\newcommand{\CS}{{$\texttt{CS}$}\xspace}
\newcommand{\CSET}{{$\texttt{CSET}$}\xspace}
\newcommand{\IMPR}{{$\texttt{IMPR}$}\xspace}
\newcommand{\JSUB}{{$\texttt{JSUB}$}\xspace}
\newcommand{\BS}{{$\texttt{BS}$}\xspace}
\newcommand{\SumRDF}{{$\texttt{SumRDF}$}\xspace}
\newcommand{\NRP}{{\sl NRP}\xspace}
\newcommand{\DeepDB}{{\sl DeepDB}\xspace}
\newcommand{\Naru}{{\sl Naru}\xspace}
\newcommand{\DBEst}{{\sl DBEst}\xspace}
\newcommand{\CONF}{{$\texttt{CON}$}\xspace}
\newcommand{\MAR}{{$\texttt{MAR}$}\xspace}
\newcommand{\ENT}{{$\texttt{ENT}$}\xspace}
\newcommand{\CTC}{{$\texttt{CTC}$}\xspace}
\newcommand{\RAND}{{$\texttt{RAN}$}\xspace}
\newcommand{\ENS}{{$\texttt{ENS}$}\xspace}
\newcommand{\GQL}{{$\texttt{GQL}$}\xspace}
\newcommand{\ORI}{{$\texttt{ORI}$}\xspace}
\newcommand{\GFlow}{{$\texttt{GFlow}$}\xspace}

\newcommand{\imdb}{\kw{IMDB}}
\newcommand{\forest}{\kw{forest}}
\newcommand{\dblp}{\kwnospace{DBLP}}
\newcommand{\gene}{\kwnospace{GENE}}
\newcommand{\aids}{\kw{aids}}
\newcommand{\hprd}{\kw{hprd}}
\newcommand{\yeast}{\kw{yeast}}
\newcommand{\wordnet}{\kw{wordnet}}
\newcommand{\youtube}{\kw{youtube}}
\newcommand{\eu}{\kw{eu2005}}
\newcommand{\yago}{\kw{yago}}
\newcommand{\uncertain}{\varphi}
\newcommand{\uncertainLC}{\varphi_{\kw{CON}}}
\newcommand{\uncertainM}{\varphi_{\kw{MAR}}}
\newcommand{\uncertainH}{\varphi_{\kw{ENT}}}
\newcommand{\uncertainCT}{\varphi_{\kw{CTC}}}
\newcommand{\Fagg}{\phi_{{a}}}
\newcommand{\Fcom}{\phi_{{c}}}

\newcommand{\ordemb}{\kw{OrderEmbedding}}
\newcommand{\erm}{\kw{ERM}}

\newcommand{\loss}{\mathcal{L}}
\newcommand{\mse}{\mathcal{L}_{mse}}
\newcommand{\batch}{\mathcal{B}}
\newcommand{\relu}{\mathsf{ReLU}}
\newcommand{\softmax}{\mathsf{softmax}}
\newcommand{\COUNT}{\kw{cnt}}
\newcommand{\AVG}{\kw{avg}}
\newcommand{\SUM}{\kw{sum}}
\newcommand{\MIN}{\kw{min}}
\newcommand{\MAX}{\kw{max}}
\newcommand{\dsb}{\kw{dsb}}
\newcommand{\job}{\kwnospace{JOB}-\kw{light}}

\newcommand{\tbd}{{\color{green}TBD}}

\newcommand{\Model}{\mathcal{M}}
\newcommand{\Dataset}{\mathcal{Q}}
\newcommand{\Prob}{\mathcal{P}}

\newcommand{\Real}{\mathbb{R}}

\newcommand{\Qerror}{\kwnospace{q}\mbox{-}\kw{error}}

\newcommand{\PyTorch}{{\sl PyTorch}\xspace}
\newcommand{\Accelerate}{{\sl Accelerate}\xspace}

\newcommand{\norm}[1]{\left\lVert#1\right\rVert}

\newcommand{\ccross}{\ding{56}}

\newcounter{example}
\renewcommand{\theexample}{\arabic{example}}

\newcommand{\score}{\ensuremath{\text{score}}\xspace}
\newcommand{\SelectFromBatch}{\textit{SelectFromBatch}\xspace}

\renewcommand{\shortauthors}{Trovato et al.}

\begin{abstract}
Large language models (LLMs) are increasingly used for semantic query processing over large corpora. A set of semantic operators derived from relational algebra has been proposed to provide a unified interface for expressing such queries, among which the semantic filter operator serves as a cornerstone. Given a table $T$ with a natural language predicate $e$, for each tuple in the relation, the execution of a semantic filter proceeds by constructing an input prompt that combines the predicate $e$ with its content, querying the LLM, and obtaining the binary decision. However, this tuple-by-tuple evaluation necessitates a complete linear scan of the table, incurring prohibitive latency and token costs. Although recent work has attempted to optimize semantic filtering, it still does not break the linear LLM invocation barriers. To address this, we propose Clustering-Sampling-Voting (CSV), a new framework that reduces LLM invocations to sublinear complexity while providing error guarantees. CSV embeds tuples into semantic clusters, samples a small subset for LLM evaluation, and infers cluster-level labels via two proposed voting strategies: UniVote, which aggregates labels uniformly, and SimVote, which weights votes by semantic similarity. Moreover, CSV triggers re-clustering on ambiguous clusters to ensure robustness across diverse datasets. The results conducted on real-world datasets demonstrate that CSV reduces the number of LLM calls by 1.28-355$\times$ compared to the state-of-the-art approaches, while maintaining comparable effectiveness in terms of Accuracy and F1 score.
\end{abstract}

\maketitle

\pagestyle{\vldbpagestyle}
\begingroup\small\noindent\raggedright\textbf{PVLDB Reference Format:}\\
\vldbauthors. \vldbtitle. PVLDB, \vldbvolume(\vldbissue): \vldbpages, \vldbyear.\\
\href{https://doi.org/\vldbdoi}{doi:\vldbdoi}
\endgroup
\begingroup
\renewcommand\thefootnote{}\footnote{\noindent
This work is licensed under the Creative Commons BY-NC-ND 4.0 International License. Visit \url{https://creativecommons.org/licenses/by-nc-nd/4.0/} to view a copy of this license. For any use beyond those covered by this license, obtain permission by emailing \href{mailto:info@vldb.org}{info@vldb.org}. Copyright is held by the owner/author(s). Publication rights licensed to the VLDB Endowment. \\
\raggedright Proceedings of the VLDB Endowment, Vol. \vldbvolume, No. \vldbissue\ %
ISSN 2150-8097. \\
\href{https://doi.org/\vldbdoi}{doi:\vldbdoi} \\
}\addtocounter{footnote}{-1}\endgroup

\ifdefempty{\vldbavailabilityurl}{}{
\vspace{.3cm}
\begingroup\small\noindent\raggedright\textbf{PVLDB Artifact Availability:}\\
The source code, data, and/or other artifacts have been made available at \url{https://github.com/Anto-an/CSV_SemanticFilter}.
\endgroup
}

\section{Introduction}

\begin{figure}
    \begin{tabular}[h]{c}
    \subfigure[\lotus on IMDB-Review]{
        
         \includegraphics[width=0.45\columnwidth]{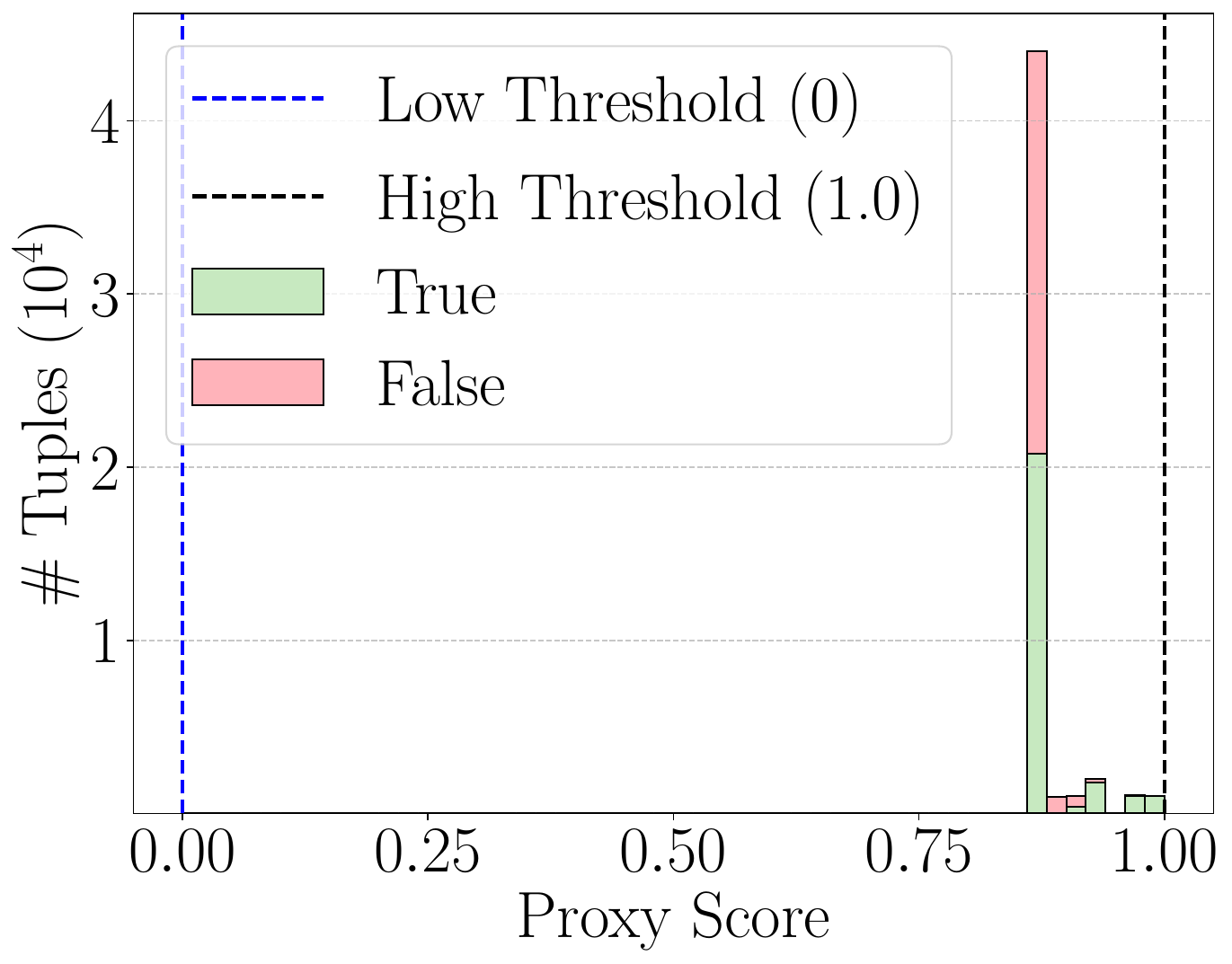}
         \label{fig: lotus_review}
    }
    
    \subfigure[\lotus on Airdialogue]{
        \includegraphics[width=0.45\columnwidth]{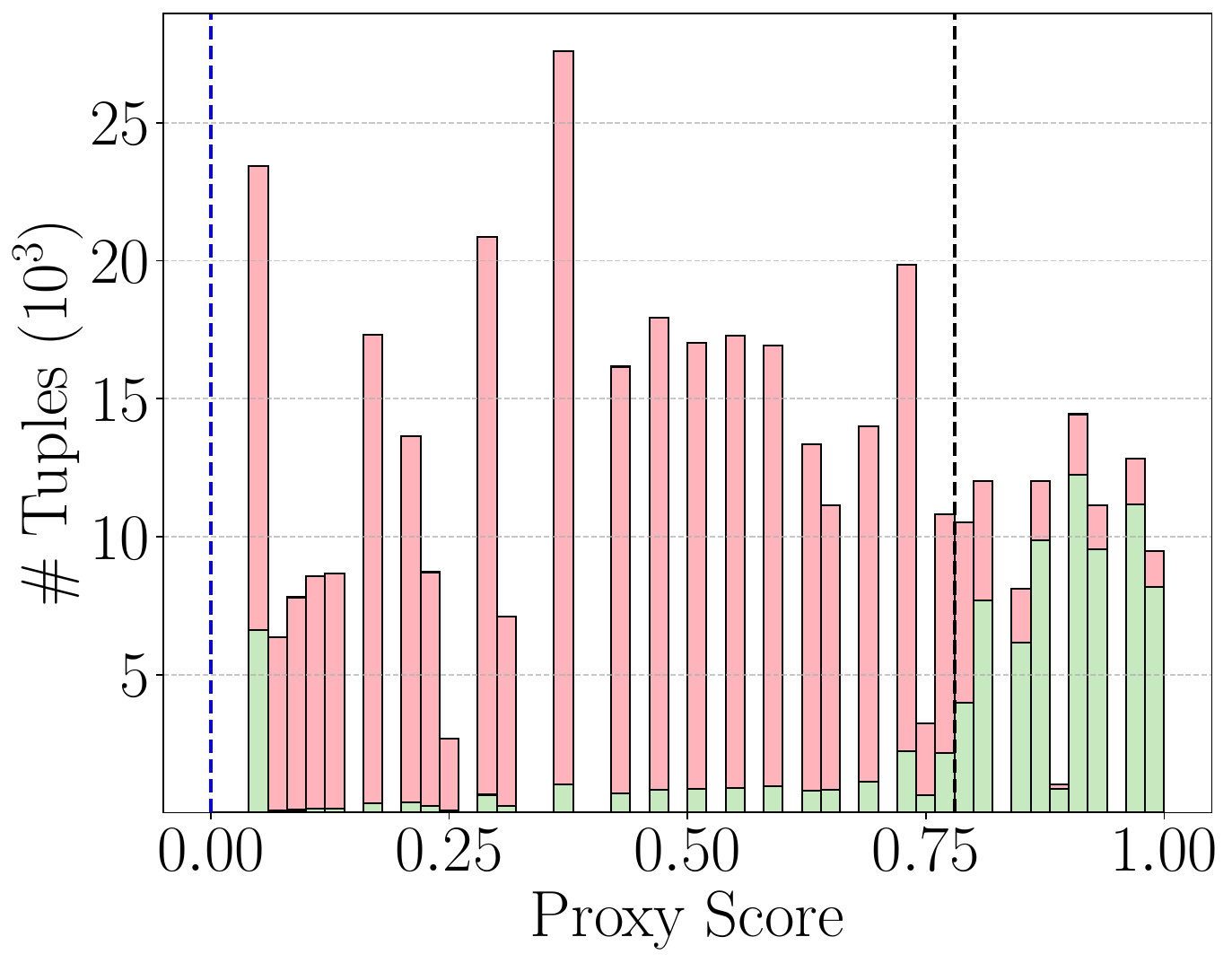}
        \label{fig: lotus_airdialogue}
        
    }
    \\
    \subfigure[\CSV on IMDB-Review]{
        \includegraphics[width=0.45\columnwidth]{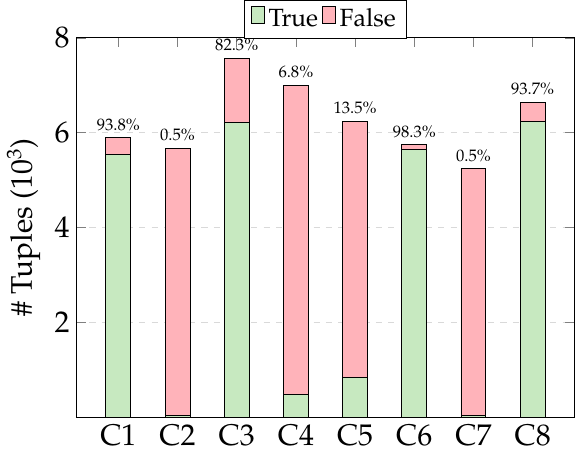}
        \label{fig:csv_review}
    }
    \subfigure[\CSV on Airdialogue]{
        \includegraphics[width=0.45\columnwidth]{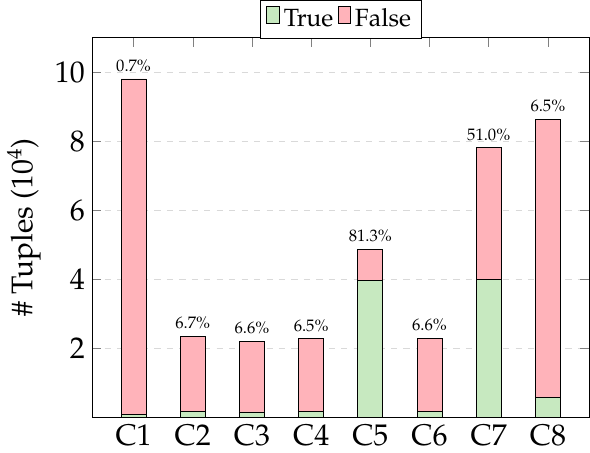}
        \label{fig:csv_airdialogue}
    }
    \end{tabular}
    \vspace{-3mm}
    \caption{Data Label Distributions of Lotus and Ours (CSV)}
    \label{fig:intro}
\end{figure}

Large language models (LLMs) have rapidly become essential tools in modern data processing, demonstrating an unprecedented capacity to understand and reason over natural language queries using their embedded world knowledge.
A wide range of complex analytical applications have been empowered by LLMs, including multi-document summarization~\cite{zhang2020pegasus}, causal relation extraction~\cite{chi2024unveiling}, and legal contract analysis~\cite{chalkidis2021lexglue}.
As semi-structured and unstructured data continue to grow, there is an increasing need for analytical systems to support bulk semantic query processing, i.e., evaluating a large volume of data through LLM-based semantic transformations. Unlike classical relational operators that follow symbolic reasoning, semantic query processing is oriented towards pragmatics and relies on context, shared knowledge, and consensus for reasoning.

To bridge the gap between semantic query processing and scalable data systems, many recent systems~\cite{DBLP:conf/icde/WangF25, DBLP:journals/corr/abs-2410-12189, DBLP:journals/corr/abs-2409-00847, DBLP:journals/corr/abs-2405-04674, DBLP:journals/corr/abs-2504-14764, dai2024uqe, liu2025palimpzest, patel2024lotus, russo2025abacus} have been proposed that enable users to specify semantic queries via DataFrame APIs or extended SQL with user-defined functions (UDFs). 
Despite differences in query expression and underlying system design, a coreset of semantic transformation operators can be `distilled' from the queries they support.
Analogous to the relational operators, \lotus~\cite{patel2024lotus} formulates a set of semantic operators, including \textit{semantic filter, semantic map, semantic join}, etc., forming a declarative and composable query interface that enables operator-level algorithmic optimization. 
Among these operators, the semantic filter, serving as the counterpart of selection in relational algebra, is the cornerstone of queries in almost all systems.
Given a table $T$ with one or more textual attributes and a natural language (NL) predicate $e$, the semantic filter returns all tuples in $T$ whose textual attributes satisfy the predicate $e$ by LLM judgment. 
Users can use this operator to express a binary classification intrinsically. For example, in the IMDB-Review dataset~\cite{IMDB-review}, the semantic filter with the NL predicate \textit{`the review is positive'} can be used for sentiment analysis. 
Distinguished from relational operator evaluation, semantic operators combine the context of the NL predicate with the input tuple, together with the LLM's underlying knowledge learned from large corpora. Consequently, the outputs are inherently uncertain with respect to prompt structure, LLM backbone, and model configurations.

The conventional approach adopted by all existing systems for semantic filtering is to scan the table $T$ linearly and issue an LLM invocation for each tuple. When the table is large, this approach incurs high query latency and substantial token consumption. 
To improve the query efficiency of the semantic filter, \lotus~\cite{patel2024lotus} introduces a two-stage model cascading algorithm, taking into account that LLM inference constitutes the bottleneck of query processing. 
The algorithm employs a lightweight proxy model, typically a small-sized LLM, to prefilter tuples and collect a fraction of tuples to learn thresholds for forwarding uncertain tuples to a more powerful LLM for final verification.
However, in practice, we find that the cascading algorithm encounters dilemmas regarding both query accuracy and efficiency on real-world datasets.
As illustrated in Fig.s~\ref{fig: lotus_review} and \ref{fig: lotus_airdialogue}, all tuples with proxy scores between the low and high thresholds must be forwarded to the powerful LLM. 
In \lotus, even in the best case, all data must be processed linearly by the proxy model. In the worst case, a poorly calibrated threshold causes nearly all data to be passed to the powerful LLM, negating the intended efficiency gains and exacerbating the overall inference cost.
This issue is critical; however, optimization for the fundamental semantic filter operator has unfortunately been neglected.

To this end, we are motivated to reconsider optimizations for processing semantic filters at an algorithmic level, where a question arises: \textbf{can we reduce the complexity of LLM invocations while providing error guarantees for semantic filters?}
For relational filtering, database systems construct B+-tree or Bitmap indices on the query attributes, replacing full table scans with local index lookups.
For semantic filtering, although existing semantic analysis systems~\cite{liu2023} allow users to build semantic indices using various pre-trained embedding models, it remains unclear whether the vectorized semantic indices can be leveraged for algorithm speedup and how to use them to reduce LLM invocations in semantic filtering.
Furthermore, theoretical error analysis, which is highly intertwined with algorithm design, is more challenging but critical.
Users typically treat LLMs as black boxes due to their sophisticated learning mechanisms, making query accuracy guarantees a necessary concern for system and algorithm designers. 
Despite the opacity of model belief from the deep learning perspective, a rigorous error guarantee from the filtering algorithm perspective helps interpret the query results and guide the configurations of the algorithm and system.

Our approach for optimizing semantic filtering is first inspired by recent progress in LLM inference for batch processing~\cite{liu2025semantic, mohandoss2024context, regmi2024gpt}. Given the high inference cost of LLMs, recent research has explored input caching strategies to accelerate online queries, where historical input queries together with their inference results are persisted in a cache. 
When a new query arrives, its semantically similar queries are retrieved from the cache, and their cached results are used directly as the proxy output. 
The intuition behind input caching for inference acceleration is that the more similar the input query, the more similar its output from LLMs. 
For a semantic filter query, the input of its LLM invocation is composed of three parts: a system prompt (i.e., a uniform instruction for semantic filtering), the predicate expression, and the value of a tuple, where the system prompt and predicate expression are identical for all the tuples for a given filter predicate.
It is the tuple value that ultimately determines the results of LLM invocations, and the semantic similarity can be captured by clustering approaches.

In this paper, we propose a new paradigm dubbed \emph{Clustering-Sampling-Voting (CSV)} for LLM-powered semantic filter processing, aiming to reduce the number of LLM invocations while preserving accuracy. The CSV framework is based on the observation that semantically similar inputs tend to elicit consistent outputs from LLMs. 
By exploiting this property, CSV amortizes inference costs across similar tuples while maintaining high labeling fidelity.
The framework consists of three key phases:
\ding{172} Clustering: All tuples in a relation are embedded using a pretrained encoder and grouped into semantically similar clusters using algorithms such as K-means.
This step is performed offline and can be reused across queries.
\ding{173} Sampling:
Within each cluster, a small sample of tuples is drawn and evaluated with the LLM as representatives.
\ding{174} Voting:
For the remaining tuples in each cluster, labels are inferred based on the LLM results from the sampled set, using either majority voting (Uniform Voting) or similarity-weighted voting (Similarity-based Voting).
To ensure correctness, CSV provides theoretical guarantees that bound the discrepancy between the voting result and the expected LLM output. This is achieved by constraining the label frequency distribution within each cluster: 
if a cluster lacks semantic purity (i.e., labels are not well-separated), the framework recursively re-clusters the set of data into subsets until, in the worst case, it falls back to direct LLM evaluation for ambiguous subsets.
Furthermore, our theoretical analysis explicitly bridges the expected error with the sample ratio, offering an intuitive way to balance theoretical assurance and LLM invocations.
As illustrated in Fig.s~\ref{fig:csv_review} and \ref{fig:csv_airdialogue}, the label distributions across eight clusters in two datasets demonstrate that most clusters are dominated by a single label, validating the effectiveness of the clustering step. In rare cases, such as Cluster 7 in Fig.~\ref{fig:csv_airdialogue}, where label purity is low, the framework triggers online re-clustering to preserve accuracy. This adaptive mechanism makes CSV both flexible and robust across diverse datasets and predicates.

The contributions of this paper are summarized as follows.

\begin{enumerate}[leftmargin=*]
    \item \textbf{Algorithm Development.} We devise a new algorithm for LLM-powered semantic filtering that reduces the scale of LLM invocations to sublinear in the average case. 
    \item \textbf{Theoretical Analysis.} We provide a detailed theoretical analysis of the proposed algorithm, which explicitly connects the error bound to the sample ratio.
    \item \textbf{Experimental Validation.} We conduct substantial experiments on multiple datasets and benchmarks to demonstrate the effectiveness and efficiency of our approaches. Our results show significant improvements in execution time, LLM call reduction, and token consumption compared to existing methods.
\end{enumerate}

\section{Preliminaries \& Existing Approaches}
\label{sec:preliminary}

A semantic operator is a declarative transformation on relational data, which is parametrized by a natural language expression as a predicate. In this paper, we consider the semantic filter powered by Large Language Models (LLMs). 

\subsection{Problem Statement}

We use $T(A_1, \cdots, A_j)$, or $T$ for short, to denote a table, where $A_1, \cdots, A_j$ are the attributes of $T$. $t \in T$ denotes a tuple in the table $T$ and $t(A_j)$ denotes the value of the attribute $A_j$ in the tuple $t$. $M: \mathcal{T} \mapsto \mathcal{T}$ is a pre-trained LLM that accepts input and generates output in textual space $\mathcal{T}$.

\stitle{Semantic Filter.}
Given a natural language expression $e$,  a semantic filter $\sigma_M$ over one or multiple textual attributes of $T$ by an LLM $M$, analogous to the relational selection operator, returns all the tuples in $T$ that satisfy the semantic predicate defined by $e$. In the operator of Eq.~\eqref{eq:sem_filter}, $M(t, e)$ denotes one LLM invocation, where the involved attributes in the tuple $t$, together with the expression $e$ and an instruction, comprise the input of LLM, prompting the LLM to determine whether or not $t$ can pass the semantic predicate $e$. 
\begin{align}
    \sigma_{M(e)}(T) = \{t \in T \mid M(t, e) = \texttt{True}\}.
    \label{eq:sem_filter}
\end{align}

\subsection{Existing Approaches and Issues}

There are three existing approaches to semantic filtering in the algorithm aspect~\cite{patel2024lotus, liu2025palimpzest, DBLP:journals/pacmmod/ZeighamiSP25}, with the exclusion of discussing other approaches due to their focus on system-level optimization over algorithmic enhancements~\cite{liu2025palimpzest,russo2025abacus}.

\stitle{Reference Algorithm~\cite{patel2024lotus, liu2025palimpzest, DBLP:journals/pacmmod/ZeighamiSP25}:}
A straightforward algorithm to process a semantic filter is to invoke one LLM call $M(t, e)$ for each $t \in T$ sequentially or in parallel via table scan. The algorithm entails $\bigo(|T|)$ LLM invocations in total, which is a huge cost.

\stitle{\lotus~\cite{patel2024lotus}:}
To reduce the cost of per-tuple LLM evaluation, \lotus proposes a two-stage algorithm that uses a cheaper proxy LLM to avoid invoking an expensive oracle LLM on every tuple in $T$.
For a semantic filter $\sigma_{M(e)}$, \lotus first samples a fixed number of tuples from $T$ and learns two thresholds, $\tau_{+}$ and $\tau_{-}$, as decision boundaries of $\texttt{True}$ and $\texttt{False}$, respectively. 
For each tuple $t \in T$, the proxy LLM is invoked to compute a proxy score, typically derived from the log-likelihood of LLM's output, which is available from open-source LLMs or certain APIs. If the proxy score is below $\tau_{-}$, \lotus labels the result as $\texttt{False}$; if it exceeds $\tau_{+}$, the result is marked as $\texttt{True}$; 
Tuples whose proxy scores fall in the ambiguous interval $[\tau_{-}, \tau_{+}]$ are forwarded to the more powerful oracle LLM for final verification.

\noindent\underline{Issues:}
Although \lotus provides an {accuracy} guarantee~\cite{DBLP:journals/pvldb/KangGBHZ20}, its practical performance can degrade for several reasons.  First, proxy cores are not always consistently well-aligned with the true labels. In favorable cases (e.g., Fig.~\ref{fig: lotus_airdialogue} on Airdialogue), tuples spread across the score range and $\texttt{True}$ tuples tend to receive higher scores than $\texttt{False}$ tuples, making it easy to learn suitable thresholds that confidently label non-trivial fraction of tuples.
However, in common cases (e.g., Fig.~\ref{fig: lotus_review} on IMDB-Review), proxy scores concentrate in a narrow interval (i.e., $[0.8, 0.85]$), and $\texttt{True}$ tuples do not reliably score higher than $\texttt{False}$ tuples. This overlap makes threshold learning unstable from a small sample, and can yield degenerate thresholds, essentially reverting to \basefilter algorithm.
Second, the sampling procedure used to fit $\tau_{+}$ and $\tau_{-}$ can be mismatched to typical semantic filter workloads. The sampling strategy of \lotus is designed for scenarios requiring the extraction of low-selectivity tuples~\cite{DBLP:conf/cidr/KangRBKZ22}. For semantic filters where labels are closer to balanced (as in Fig.s~\ref{fig: lotus_review} and \ref{fig: lotus_airdialogue}), this strategy can produce a biased sample, which in turn drives the thresholds towards extreme values and reduces the proxy’s ability to confidently rule out negatives.
Third, the two-stage cascade algorithm does not break the bottleneck of linear LLM invocation and can even increase total cost. \lotus must invoke the proxy model for each tuple, which is still an LLM with considerable inference cost. When thresholds are poorly calibrated, a large fraction of tuples is forwarded to the oracle LLM, leading to an additional near-linear pass. In such cases, the two-stage cascade can exceed the cost of directly applying the Reference algorithm. 

\stitle{\BARGAIN~\cite{DBLP:journals/pacmmod/ZeighamiSP25}:}
Similarly, \BARGAIN begins by invoking a proxy LLM on each tuple $t \in T$ to obtain a proxy prediction and an associated proxy score. It then partitions tuples into score regions defined by pre-specified intervals (e.g., 20 score subranges). \BARGAIN starts sampling from the region with the highest proxy scores and progressively expands to lower-score regions. Within each region, it samples tuples and invokes the oracle LLM to test whether the tuples in that region can meet a user-specified quality target. The region-wise procedure terminates once \BARGAIN identifies a minimum proxy-score threshold such that including any lower-score region would violate the target quality. Finally, tuples whose proxy scores fall below this threshold are forwarded to the oracle LLM for final verification.

\noindent\underline{Issues:}
Since \BARGAIN requires scanning the dataset with a substantial proxy LLM, akin to \lotus, it needs a considerable overhead on this linear processing. Importantly, recent studies~\cite{DBLP:conf/icml/GalG16, DBLP:conf/naacl/GengCWKNG24} indicate that neural networks frequently yield poorly calibrated confidence scores, inaccurately representing true correctness probabilities due to their inherent overconfidence of neural networks. This issue could potentially impact \BARGAIN's performance, as the proxy score forms the foundation of its approach.

\section{Our Approach for Semantic Filter}
\label{sec:sem_filter}
In this section, we present a new semantic filter algorithm based on a clustering-sampling-voting (CSV) paradigm, designed to address the issues discussed in \cref{sec:preliminary}. 
At a high-level, our approach replaces per-tuple proxy LLM calling with embedding-based grouping. We compute an embedding for each tuple using an embedding model; these embeddings can be generated offline (e.g., at table creation time or upon tuple insertion) and are substantially cheaper than invoking an LLM to obtain proxy scores.
We cluster tuple embeddings and infer a label for each cluster based on whether $\texttt{True}$ or $\texttt{False}$ is dominated within the cluster. As illustrated in Fig.s~\ref{fig:intro}(c) and (d), this enables labeling a large number of tuples without an LLM invocation. 

Our CSV paradigm builds on the intuition that embedding-similar tuples tend to elicit similar responses from LLMs; thereby, embedding-based grouping can serve as an initial proxy.  This intuition is commonly exploited in inference caching: when two prompts are close in semantic space, their outputs are likely to be similar in meaning, enabling the reuse of prior inference results~\cite{liu2025semantic, mohandoss2024context, regmi2024gpt}.  In semantic filtering, the system prompt and predicate expression are fixed, and only the tuple content varies across invocations. Consequently, variability in the LLM's output is fully driven by the tuple, and semantic proximity between tuples can be captured by distances in an embedding space.

We empirically validate this assumption by quantifying the relationship between embedding distance and label agreement across multiple predicates. 
For each tuple, we compute its embedding using E5-Large~\cite{wang2022e5} and obtain from the LLM the probability of the $\texttt{True}$ label.
For each tuple pair $(t_i, t_j)$, we compute the Euclidean distance between their embeddings and the probability that the two tuples share the same label of $\texttt{True}$ or $\texttt{False}$. 
By splitting the range of all-pair distances into 50 bins, Fig.~\ref{fig:curve} illustrates the trends of the mean probability of label agreement as the distances increase across 6 queries on two datasets.
Across predicates, the results show a consistent pattern: as embedding distance increases, the probability of label agreement decreases. The magnitude of the effect varies by predicate—e.g., RV-Q1, RV-Q2, and CB-Q2 exhibit a sharper decline, whereas RV-Q3 and CB-Q1 decrease more gradually—but the overall trend remains stable. These findings provide quantitative support for using embedding-based clustering as a mechanism to amortize LLM evaluations in semantic filtering. In the following, we present the details of CSV in \cref{sec:cluster-vote} and provide its theoretical analysis in \cref{sec:theo}. 

\begin{figure}[t]
\begin{tabular}{c}
     \subfigure[RV-Q1]{
     \includegraphics[width=0.32\columnwidth]{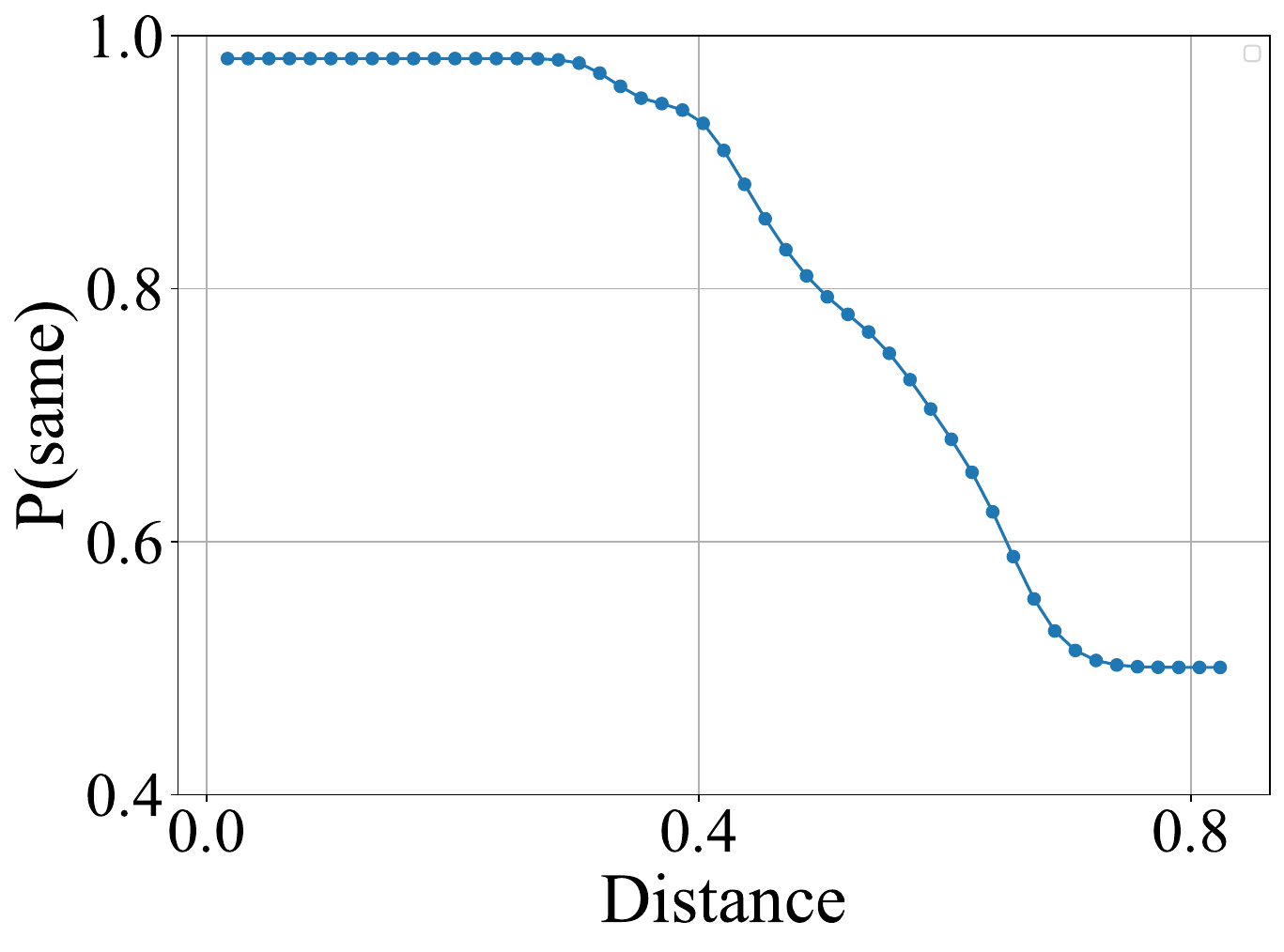}
     \label{fig:curve_RV_Q1}
     }
     \hspace{-0.3cm}
     \subfigure[RV-Q2]{
        \includegraphics[width=0.32\columnwidth]{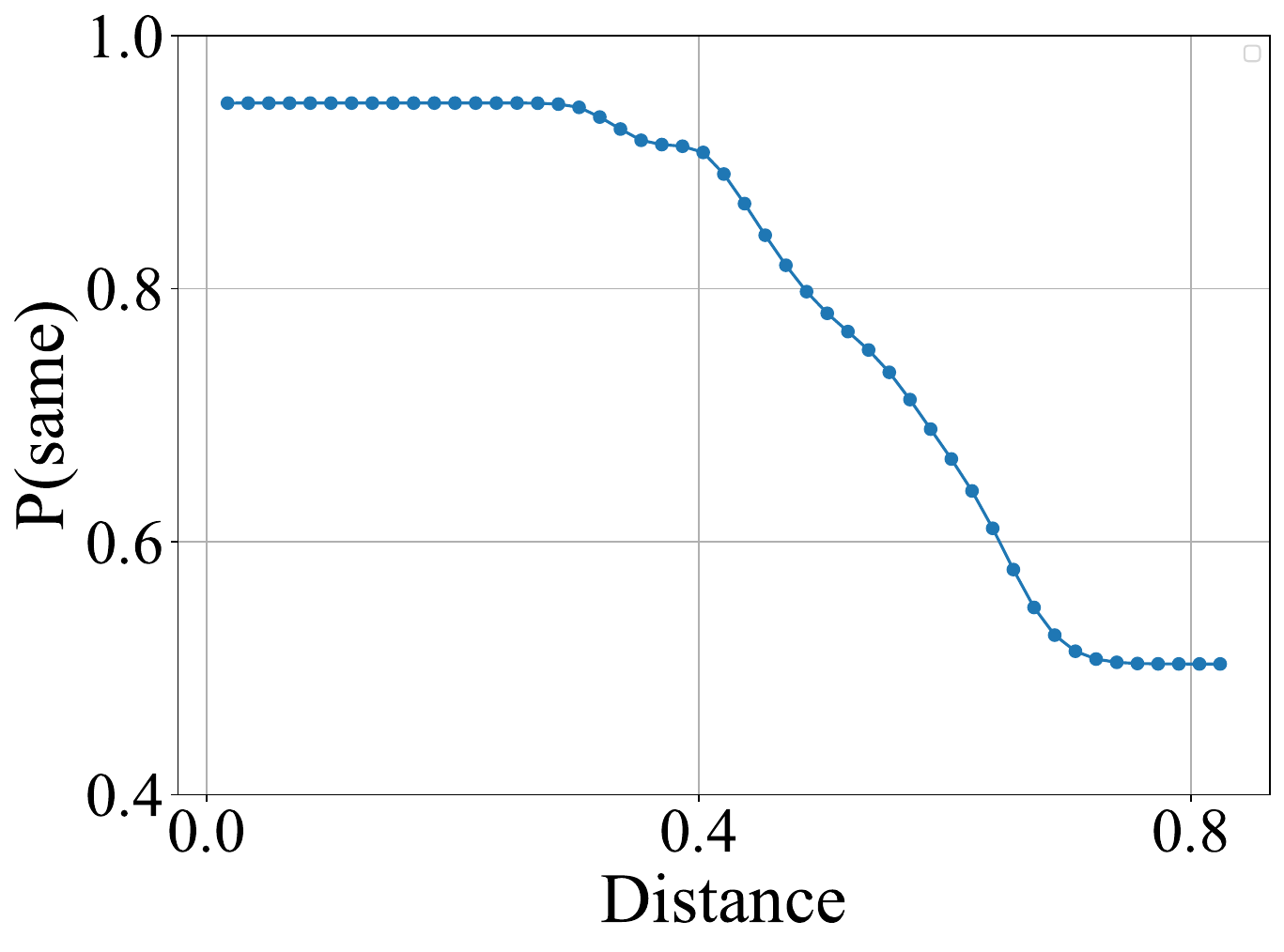}
        \label{fig:curve_RV_Q2}
     }
     \hspace{-0.3cm}
     \subfigure[RV-Q3]{
        \includegraphics[width=0.32\columnwidth]{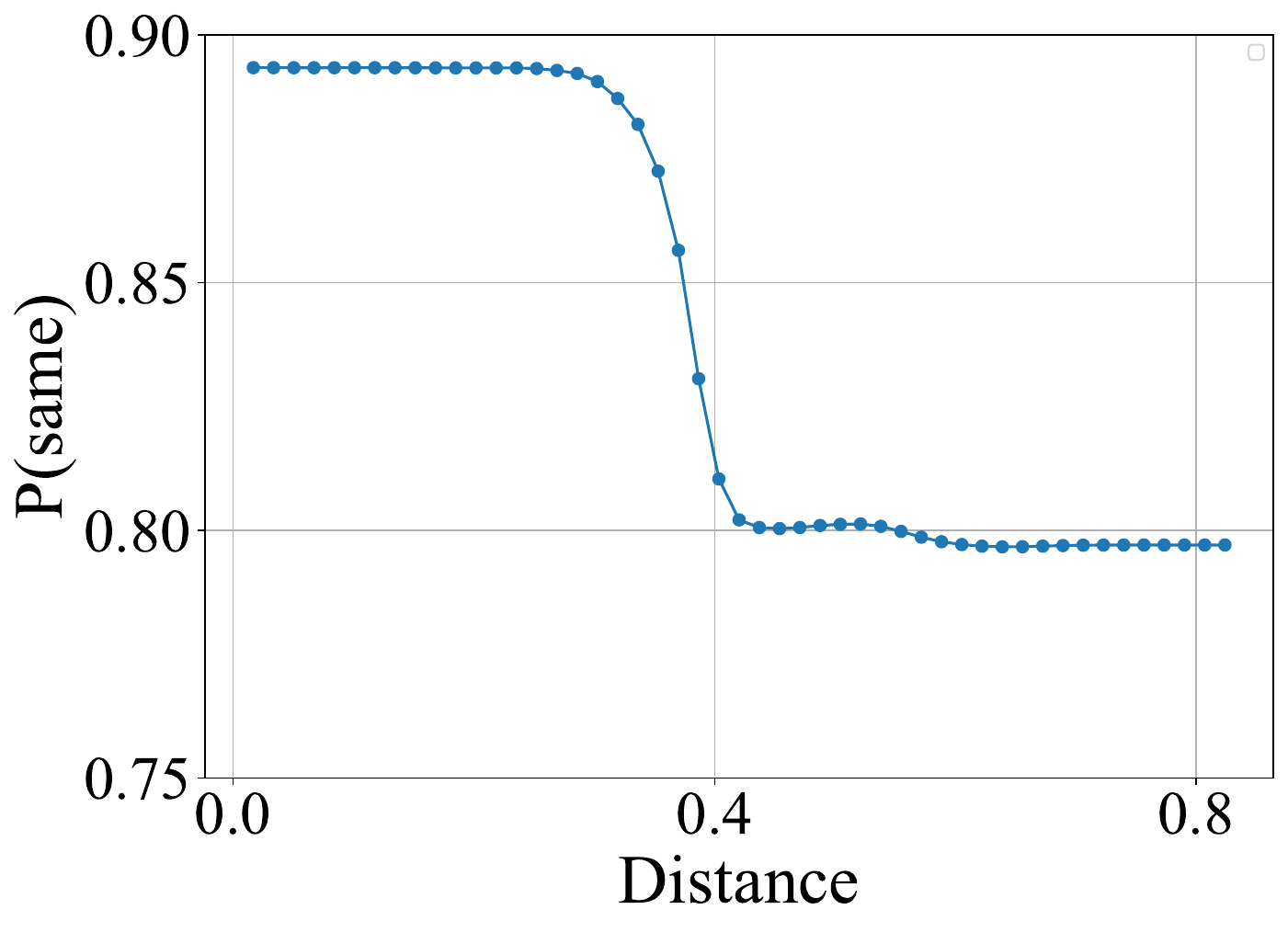}
        \label{fig:curve_RV_Q3}
     }
     
\end{tabular}
\vspace{-0.5mm}
\begin{tabular}{c}
     \subfigure[CB-Q1]{
        \includegraphics[width=0.32\columnwidth]{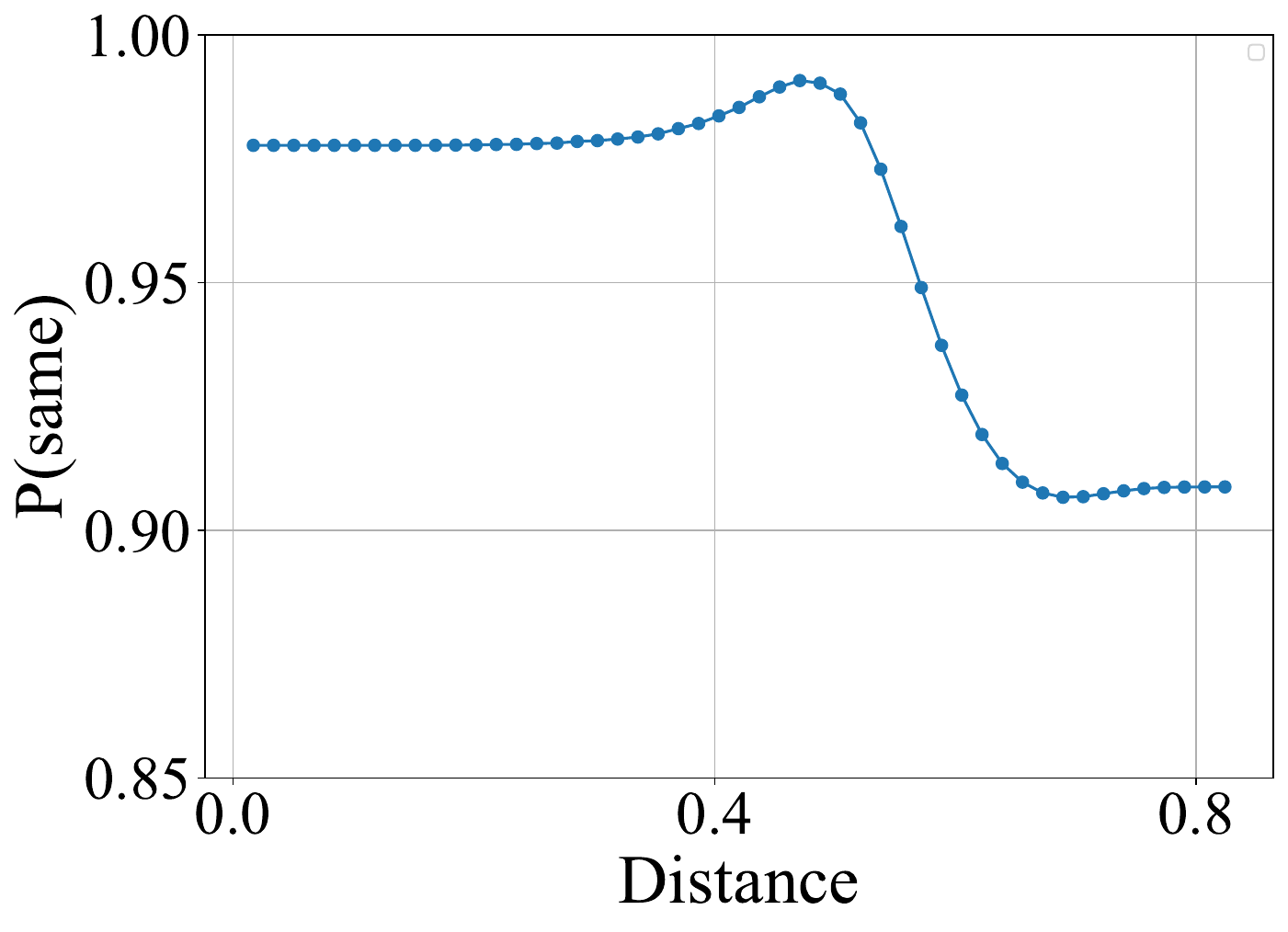}
        \label{fig:curve_CB_Q1}
     }
     \hspace{-0.3cm}
     \subfigure[CB-Q2]{
        \includegraphics[width=0.32\columnwidth]{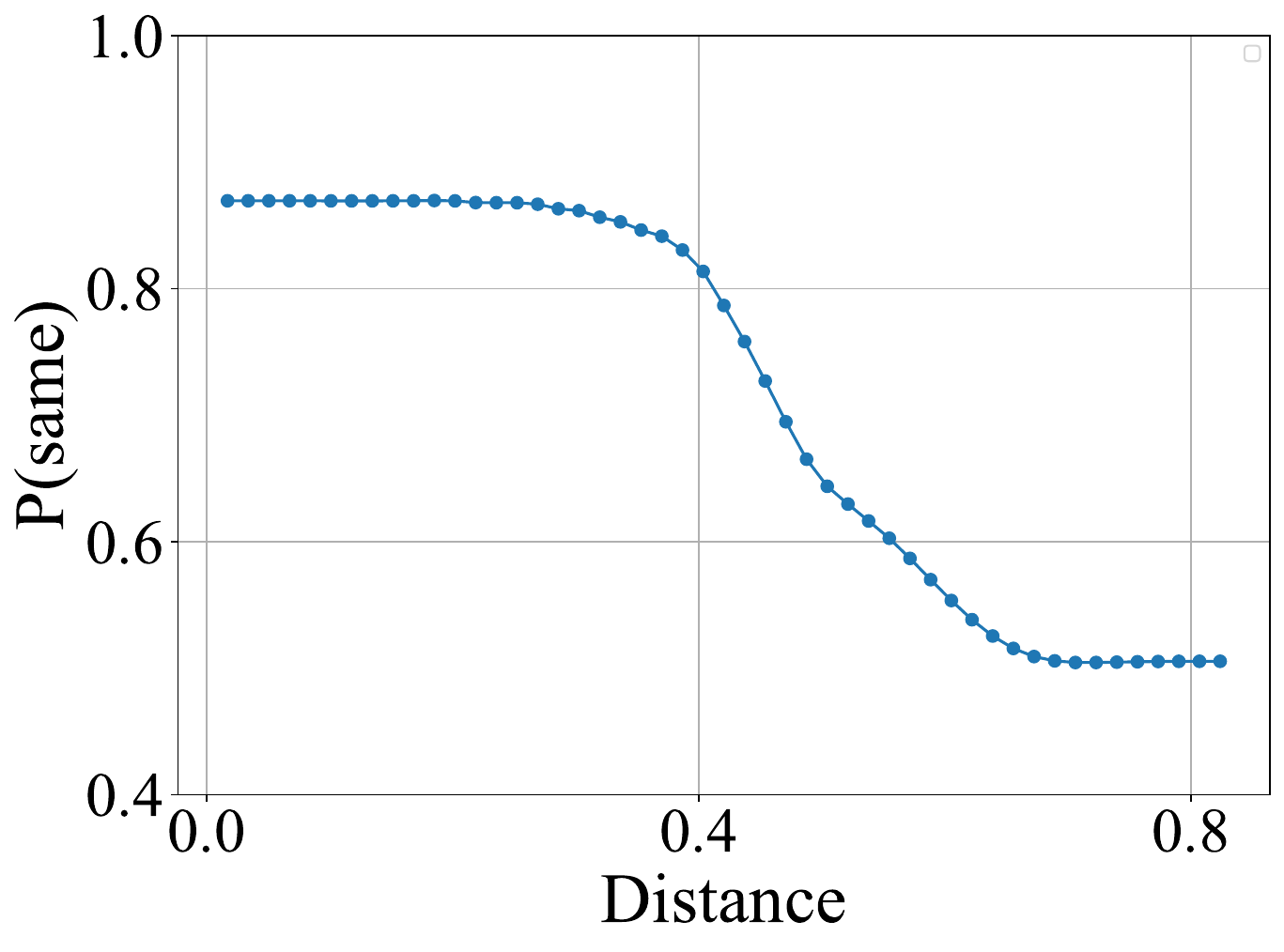}
        \label{fig:curve_CB_Q2}
     }
     \hspace{-0.3cm}
     \subfigure[CB-Q3]{
        \includegraphics[width=0.32\columnwidth]{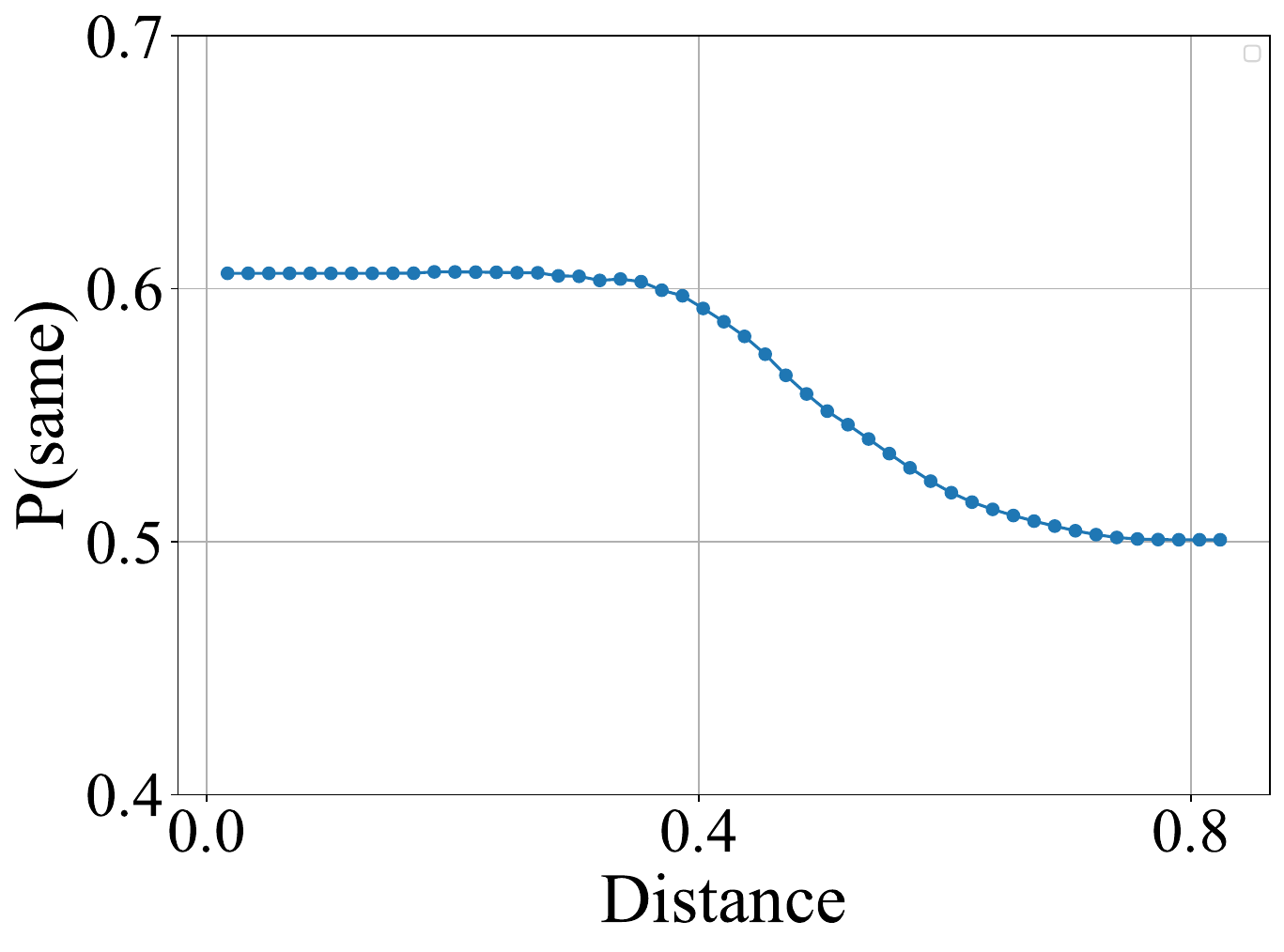}
        \label{fig:curve_CB_Q3}
     }

\end{tabular}
\vspace{-0.5mm}
     \caption{Probability Distribution Across Distance}
     \label{fig:curve}
\end{figure}

\subsection{A Clustering-Sampling-Voting Paradigm}
\label{sec:cluster-vote}

As this name indicates, the Clustering-Sampling-Voting (CSV) paradigm processes semantic filter queries in two phases: an offline clustering phase and an online sampling with voting phase. 
In the offline phase, we employ a clustering algorithm such as K-means to partition the tuples in the table $T$ into clusters, which is query-agnostic. 
In the online phase, we sample a subset of tuples from each cluster, invoke LLM to evaluate these sampled tuples regarding the semantic predicates, and ensemble their results via voting to determine the remaining tuples in the cluster. As the voting mechanism will directly bypass LLM invocations for unsampled tuples, to preserve high accuracy, we use two thresholds $ub$ and $lb$ to limit this direct filtering for tuples with low confidence. Also, in the following theoretical analysis part, we prove that the error is guaranteed by thresholds $ub$ and $lb$.
The algorithm initiates a fallback mechanism for these low-confident tuples, resorting to online fine-grained re-clustering, re-sampling, and re-voting to refine decisions. If uncertainty persists and the recursive depth reaches its maximum threshold, the framework ultimately falls back to direct LLM invocation for the remaining ambiguous tuples.

\begin{algorithm}[t]
    \caption{\textsc{SemanticFilter}($T$, $e$, $M$, $k$, $\xi$)}
    \label{alg: semantic_filter}
    \SetKwData{Up}{up}  \SetKwInOut{Input}{Input} \SetKwInOut{Output}{Output}
	\Input{input table $T$, expression $e$, LLM $M$, \#clusters $k$, and sample ratio $\xi$}
	\Output{output table $R$}
    Initialize $R \gets \emptyset$ \; 
    Generate the embedding $\mathbf{e}_i$ for each tuple $t_i \in T$ \; 
    Partition $\{(\mathbf{e}_1, t_1), \cdots, (\mathbf{e}_n, t_n)\}$ into $k$ clusters $\mathcal{C} = \{ C_1, \cdots, C_k\}$ by $k$-means over $\{ \mathbf{e}_1, \cdots, \mathbf{e}_n\}$\; 
    Initialize $Q \gets {\mathcal{C}}$ \;
    \While{$Q$ is not empty}{
        Initialize $\mathcal{U} \gets \emptyset$ \;
        \For{$C_j \in Q$} {
            $Q.\text{remove}(C_j)$ \;
            Randomly sample subset $C'_j \subset C_j$ with ratio $\xi$ \;
            Allocate $O \gets \emptyset$ \;
            \For{$(\mathbf{e}_i, t_i) \in C'_j$} {
                $o_i \gets M(t_i, e)$ \;
                $O \gets O \cup \{o_i, t_i, \mathbf{e}_i)\}$ \;
                \If{$o_i$ is \texttt{True}}{
                    $R \gets R \cup \{t_i\}$ \;
                }
            }
            $C^*$, $R^+$ $\gets$ \textsc{Vote}($O$, $C_j$, $C'_j$) \;
            $R \gets R \cup R^{+}$ \;
            \If{$C^*$ is not empty} {
                $\mathcal{U} \gets \mathcal{U} \cup C^*$ \;
            }
        }
        \If{$C^*$ is not empty} {
            Partition $\mathcal{U}$ into new clusters $\{C_{j_1}, \cdots, C_{j_m}\}$ by $k$-means \;
            $Q \gets \{C_{j_1}, \cdots, C_{j_m}\}$ \;
        }
        
    }
    Return $R$ \;
\end{algorithm}

Algorithm~\ref{alg: semantic_filter} presents the procedure of CSV for the semantic filter, given a table $T$, a semantic expression $e$, an LLM $M$, and a parameter for the number of clusters $k$. In the offline phase, the algorithm adopts an embedding model to encode each tuple $t_i$ into a vector representation $\mathbf{e}_i$ (line 2), and then partitions the tuples into $k$ clusters $\{C_1, \cdots, C_k\}$ by K-means regarding the semantic similarity of the embeddings (line 3).
In the online sampling and voting phase, we process each cluster individually by voting strategy. Here, we use a set $Q$ to maintain all the clusters to be processed. For each cluster layer, we initialize a temporary set $\mathcal{U}$ (line 8) to collect all low-confidence tuples that require further refinement.
For a cluster $C_i$, we first draw a collection of tuples, $C'_i$, uniformly with a sampling rate $\xi$ (line 9). Afterwards, we invoke the LLM $M$ to infer the result for each sampled tuple, respectively, regarding the semantic predicate $e$ (lines 10-15), where the results are persisted in a temporary set $O$. Sampled tuples that pass the LLM-based semantic filter are directly inserted into the result table $R$ (lines 14-15).
For the remaining tuples in $C_j \setminus C'_j$, the algorithm introduces the voting mechanism, i.e., the function $\textsc{Vote}$ in line 16, which takes the results of the sampled tuples, all tuples in the current cluster, and the sampled tuples as inputs. The function returns two sets of tuples, $C^*$ and $R^+$, which denote the low-confident tuple set that needs fallback verification and the high-confident tuple set that is determined to pass the filter by voting, respectively.
We append $R^+$ to the result table line 17. For the low-confidence tuples, we collect all such sets into $\mathcal{U}$ (line 18). After finishing the current layer, if $U$ is nonempty, we repartition $\mathcal{U}$ as a whole by K-means (line 19) and push the resulting subclusters back into the set $Q$ for the next iteration (line 20). Each new sub-cluster is then processed following the same procedure (lines 10-19), i.e., uniform sampling and voting, until a maximum split limit is reached.
We equip two voting strategies for the clustering \& voting algorithms, i.e., uniform voting (Algorithm~\ref{alg:univote}) and similarity-based voting (Algorithm~\ref{alg:simvote}). Both of these strategies use two thresholds $ub, lb \in (0, 1)$ to evaluate the voting belief. We elaborate on these two algorithms in the following.

\begin{figure*}
    \centering
    \includegraphics[width=\linewidth]{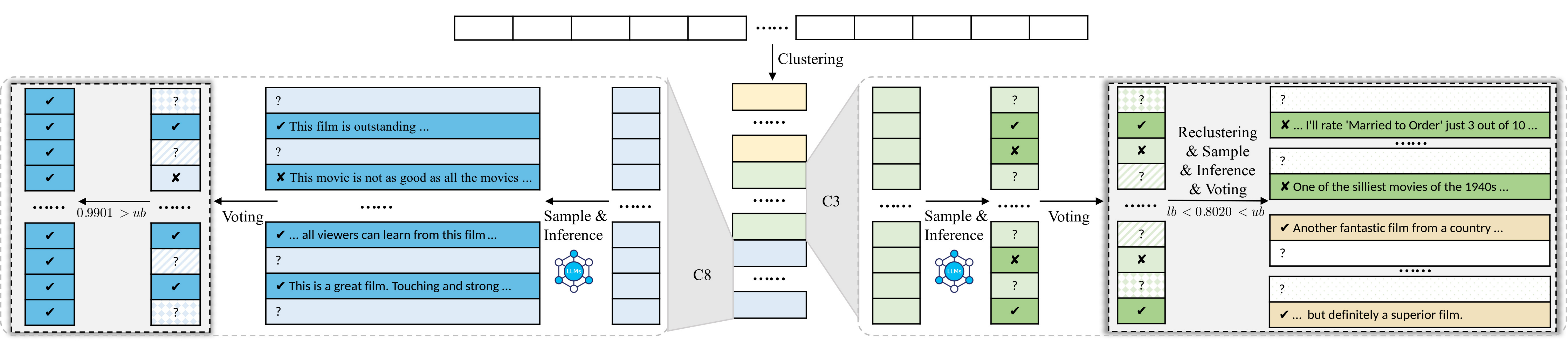}
    \caption{An Running Example of Clustering-Sampling-Voting (CSV) with \textsc{UniVote} on IMDB-Review Dataset}
    \label{fig:example:univote}
\end{figure*}

\begin{algorithm}[t]
    \caption{\textsc{UniVote}($O$, $C$, $C'$)}
    \label{alg:univote}
    \SetKwData{Up}{up}  \SetKwInOut{Input}{Input} \SetKwInOut{Output}{Output}
	\Input{a set of sampled tuples and LLM invocation results $O$, a cluster of tuples with their embedding $C$, the sample subset of tuples with their embeddings $C'$}
	\Output{undetermined tuples with their embedding $C^*$ and positive tuples $R^+$}
    $O^{+} \gets \{o_i = \texttt{True} \mid (o_i, t_i, \mathbf{e}_i) \in O \}$ \;
    \text{score} $\gets \frac{|O^{+}|}{|O|}$ \;
    \If{$\text{score} \ge \text{upperbound}$}{
        $C^* \gets \emptyset$, $R^+ \gets \{t_i \mid  (t_i, \mathbf{e}_i) \in C \setminus C' \}$ \;
    }
            
    \ElseIf{$\text{score} \le \text{lowerbound}$} {
        $C^* \gets \emptyset$, $R^+ \gets \emptyset$ \; 
    }
    \Else {
        $C^* \gets C \setminus C'$, $R^+ \gets \emptyset$ \; 
    }
    return $C^*, R^+$ \;
\end{algorithm}

\stitle{Uniform Voting.} Given the results of the sampled tuples $O$, Uniform Voting computes the ratio of tuples that pass the filter (line 1 in Algorithm~\ref{alg:univote}). Comparing it with the two thresholds $ub$ and $lb$ results in three cases as shown in Eq.~\eqref{eq:univote:proxy}.

\begin{align}
\text{score}(t_i) &= \frac{|O^+|}{|O|}, \forall t_i \in C \setminus C', \\
\hat{M}(t_i; e) &= \begin{cases}
        \text{True} & \text{if}~\text{score}(t_i) \ge ub \\
        \text{False} & \text{if}~\text{score}(t_i)  \le lb  \\
        \text{Undetermined} & \text{if}~\text{score}(t_i)  \in (lb, ub)
        \end{cases}.  \label{eq:univote:proxy}
\end{align}
Here, we use $\hat{M}(t, e)$ to denote the vote result on a tuple $t$ regarding $e$. Under each case, Uniform Voting will return corresponding tuples which are regarded as passing the filter, $R^+$, and the undetermined tuples $C^*$ (lines 3-8).
In Uniform Voting, the result of sampled tuples in $C'$ contributes equally, and the voting results $\hat{M}(t; e)$ are equal for all the remaining tuples in the cluster. 
To achieve fine-grained decision making, we further propose a similarity-based voting mechanism that takes the embedding similarity between sampled tuples and remaining tuples into account.

\stitle{Similarity-based Voting.} Algorithm~\ref{alg:simvote} presents the process of Similarity-based Voting, which also returns the two sets $R^+$ and $C^*$, as additional results and undetermined tuples. In Algorithm~\ref{alg:simvote}, for each remaining tuple $t_i$, we aggregate the voting results of each sampled tuple $t_j$ weighted by its normalized similarity between $t_i$, and the score function is reformulated as Eq.~\eqref{eq:simvoate:score}
\begin{align}
    \text{score}(t_i) &= \sum_{t_j \in C'} \frac{\text{sim}(e_i, e_j) }{\sum_{t_k \in C'} \text{sim}(e_i, e_k)} \mathbb{I}(M(t_j, e)), \forall t_i \in C \setminus C', \label{eq:simvoate:score} 
\end{align}
where $\mathbb{I}(M(t_j, e))$ is the identity function of LLM invocations on the sampled tuples. Similarly, the score function is compared with $lb$ and $ub$ following Eq.~\eqref{eq:univote:proxy}. 

Compared to Uniform Voting, Similarity-based Voting can provide robustness when initial clustering does not perform well. Specifically, when the distribution of positive and negative labels within a cluster does not satisfy our criteria, i.e., $lb$ and $ub$, Uniform Voting will re-cluster the entire cluster, whereas Similarity-based Voting may vote successfully on certain items. This reduces the amount of unprocessed data in subsequent steps. For example, the results of AD-Q1 and AD-Q4 shown in \S~\ref{sec:exp-filter}, Similarity-based Voting outperforms Uniform Voting, in reducing LLM calls.

\begin{algorithm}[t]
    \caption{\textsc{SimVote}($O$, $C$, $C'$)}
    \label{alg:simvote}
    \SetKwData{Up}{up}  \SetKwInOut{Input}{Input} \SetKwInOut{Output}{Output}
	\Input{a set of sampled tuple and LLM invocation results $O$, a cluster of tuples with their embedding $C$, the sample subset of tuples with their embeddings $C'$}
	\Output{undetermined tuples with their embedding $C^*$ and positive tuples $R^+$}
    $C^* \gets \emptyset$, $R^+ \gets \emptyset$\;
    \For{$(t_j, \mathbf{e}_j) \in C \setminus C'$}{
        $\text{score} \gets 0$\;
        $\text{norm} \gets\sum_{t_k \in C'} \text{sim}(e_j, e_k)$\;
        \For{ $(o_i, t_i, \mathbf{e}_i) \in O$ }{
            $\text{flag} \gets 1$ if $o_i$ is \texttt{True} else 0 \;
            $\text{score} \gets \text{score} + \frac{\text{sim}(\mathbf{e}_i, \mathbf{e}_j)}{\text{norm}} \cdot \text{flag}$ \;
        }
        \If{\text{score} $>$ upperbound}{
            $R^+ \gets R^+ \cup \{ t_j\}$ \; 
        }
        \ElseIf{\text{score} $<$ upperbound $\wedge$ \text{score} $>$ lowerbound}{
            $C^* \gets C^* \cup \{ (t_j, \mathbf{e}_j)\}$ \; 
        }
    }
    return $C^*, R^+$ \;
\end{algorithm}

\begin{example}(Clustering-Sampling-Voting).
Fig.~\ref{fig:example:univote} illustrates a running example of Clustering \& Uniform Voting on IMDB-Review dataset. The dataset consists of a table of tuples, each representing a movie review. The semantic filter query applied here is ``The {review} is positive.''
We first partition the table into 8 clusters using $K$-means clustering based on the tuple embeddings. The label distributions of each cluster are depicted in Fig.~\ref{fig:csv_review}.
Next, voting is conducted independently for each cluster. We use clusters C3 (highlighted in green) and C8 (highlighted in blue) as examples. In this running example, we set {$lb=0.15$ and $ub=0.85$}.
In cluster C8, {$101$ tuples are uniformly sampled}, and their semantic filter results are obtained from individual LLM invocations, e.g., the semantic filter yields \texttt{True} for the tuple ``this file is outstanding...''. Among the sampled tuples, $99.01\%$ are labeled as \texttt{True}, surpassing the upper bound ($0.9901>ub$), leading us to categorize all tuples in this cluster as \texttt{True}.
For cluster C3, {where $101$ tuples are uniformly sampled}, $80.20\%$ are classified as \texttt{True}. Since $lb<0.8020<ub$, we apply a re-cluster strategy. This strategy results in two clusters, distinguished by green and yellow colors in the far right of Fig.~\ref{fig:example:univote}. Similarly, we sample tuples from each of these two clusters, utilize LLM for individual tuple analysis, and subsequently conduct voting.
\end{example}

\noindent\textbf{Complexity Analysis.} 
Given the table $T$ and a sampling rate of $\xi$, the CSV paradigm issues $\bigo(\xi|T|)$ LLM invocations in the best case and issues $\bigo(|T|)$ LLM invocations in the worst case.

The CSV framework can support filters on multiple columns by generating a fused semantics embedding for multiple columns. 
For example, concatenating the text of multiple columns together with the column names and feeding the long context into a pre-trained language model to produce the row-wise embedding, or indexing the embeddings for all columns by hierarchical clustering. 
In addition, our approach can be extended to support tuple insert/delete
by incorporating incremental clustering algorithms to maintain the
clusters.  
Our update handling distinguishes between small and large batches.
(1) Small, occasional updates. When the update batch is small, the cluster centroids are assumed to be unchanged. New data are inserted into the clusters with the nearest centroid, and we reuse the existing cluster-level voting outcomes. If a cluster previously had no consensus, we apply linear LLM invocations to incorporate the minor updates while keeping the overhead low.
(2) Larger, periodic updates. For larger batches of updates, we use mini-batch K-means to update the clustering. Subsequent sampling, voting, and re-clustering are fully compatible with the underlying data. We cache the LLM outcomes for sampled tuples from the original datasets for potential reuse. For samples in the updated datasets, LLM is invoked to process tuples missed in the cache. 
For deletion, deleted tuples will be marked, and clusters will be merged when the number of them reaches a specific number.

\subsection{Analysis on $\xi$ Setting of Voting Strategies}
\label{sec:theo}

\cref{sec:cluster-vote} introduces two voting strategies, \textsc{UniVote} and \textsc{SimVote}. Both strategies are governed by a key parameter, the sample ratio $\xi$, which directly impacts both the time complexity (requiring $O(\xi|T|)$ LLM invocations) and the theoretical error (as elaborated in the subsequent discussion).  
Concretely, from a given cluster, the voting strategies sample a subset of tuples, constituting a proportion $\xi$ of the entire tuples in the cluster, invoke the LLM on the sampled tuples, and use the sampled outcomes to infer the cluster-level label. 
The choice of $\xi$ induces a natural trade-off. A larger $\xi$ increases the number of LLM calls but reduces estimation error; conversely, a smaller $\xi$ reduces LLM invocations with the compromise of a higher error. 
In practice, however, users typically lack the prior knowledge needed to choose  $\xi$ to balance cost and quality. 

To make principled parameter tuning, we introduce a user-friendly parameter, $\epsilon$, representing a user-specified error tolerance in our theoretical analysis. 
In the following, we derive a theoretical relationship between $\xi$ and $\epsilon$, enabling us to select an appropriate sampling ratio $\xi$ for a given $\epsilon$, so that our CSV framework achieves the desired accuracy with high probability.

First, we introduce the cornerstone of our theoretical analysis, Bernstein Inequality~\cite{bernsteinBound} in Lemma~\ref{lemma:b-bound}, which will be frequently used subsequently.
\begin{lemma}(Bernstein Inequality~\cite{bernsteinBound})
Given $X = \{x_1, \cdots, x_n\}$ is the finite population of size $n$, where $x_1, \cdots, x_n$ are binary variables with a mean $\mu$. Let $\hat{\mu}$ be the mean of $k$ variables $c_1,\cdots,c_k$ obtained by simple random sampling without replacement from $X$. For any $\epsilon>0$, we have $$Pr\left[|\hat{\mu}-\mu|\ge \epsilon\right]\le 2 \text{exp}\left(-\frac{k\epsilon^2}{2\hat{\sigma}^2 + {2R\epsilon}/{3}} \cdot \frac{n - k}{n - 1} \right),$$
where $\hat{\sigma}$ is the sample standard deviation of $k$ sampled variables $c_1,\cdots,c_k$, i.e., $\hat{\sigma}^2=\frac{1}{k-1}\sum_{i=1}^k(c_i-\hat{\mu})^2$, $R$ is the almost sure bound for $x_i$, which means $|x_i| \le R$ almost surely.
\label{lemma:b-bound}
\end{lemma}

\stitle{Analysis of \textsc{UniVote}.}
First, we analyze the \textsc{SemanticFilter} with \textsc{UniVote}.
Here, given $X = \{x_1, \cdots, x_n\}$ of $C_j$, where $x_1, \cdots, x_n$ are binary, and $x_i=\mathbb{I}(M(t_i,e))$ for each $(\mathbf{e}_i,t_i)\in C_j$.
Let $\mu(X)$ be the mean of variables in $X$, and $f(t_i)\in \{0,1\}$ indicate whether $(\mathbf{e}_i,t_i)$ is included in $R^+$ by \textsc{UniVote}. Specifically, $f(t_i)=1$ if $|O^{+}|/|O|\ge ub$ and $f(t_i)=0$ if $|O^{+}|/|O|\le lb$, where $|O|=k$ and $\xi=|O|/|C_j|$.
By the Bernstein bound, we confirm that $f(t_i)$ is an accurate estimator of $x_i$ for each $(\mathbf{e}_i,t_i)$ when the sampling rate $\xi$ is sufficiently large.

\begin{theorem}
Given a cluster $C=\{(\mathbf{e}_1,t_1),\cdots, (\mathbf{e}_n,t_n)\}$ and parameters $lb$, $ub$, $l$ and $\epsilon$. Let $X = \{x_1, \cdots, x_n\}$ of $C$, where $x_i=\mathbb{I}(M(t_i,e))$ is binary for each $(\mathbf{e}_i,t_i)\in C$. Assume that $\xi$ satisfies $$\xi\ge \frac{1}{2}-\sqrt{\frac{1}{4}+\ln l \left({\frac{2\hat{\sigma}^2}{\epsilon^2}+\frac{2}{3\epsilon}}\right)}.$$ Then, the following inequality holds with at least $1-2l^n$ probability when $\mu(O)\le lb$ or $\mu(O)\ge ub$: $$Pr[f(t_i)\neq x_i]\le\max(lb+\epsilon,1-(ub-\epsilon)).$$
\label{theo:univote}
\end{theorem}

Theorem~\ref{theo:univote} bounds the probability that the voting decision $f(t_i)$ disagrees with the true LLM output $x_i$. The key insight is that the mean of samples $\mu(O)$ concentrates around the mean of the true population $\mu(X)$ with high probability.  The condition on $\xi$ ensures that we sample enough tuples for the Bernstein inequality to guarantee this concentration within tolerance $\epsilon$. Consequently, if voting assigns a label, the error is bounded by $\max(lb+\epsilon,1-(ub-\epsilon))$, and this guarantee holds with probability at least $1 - 2l^{n}$.

\proofsketch
Let ${\mu}(O)$ be the mean of $k$ sampled elements from $X$, i.e., $k=\xi n\ge \frac{1}{2}n-n\cdot\sqrt{\frac{1}{4}+\ln l \left({\frac{2\hat{\sigma}^2}{\epsilon^2}+\frac{2}{3\epsilon}}\right)}$ and ${\mu}(O)=|O^+|/|O|$.

For the case that ${\mu}(O)\le lb$, we have $f(t_i)=0$. Hence, $Pr[f(t_i)\neq x_i]=Pr[x_i\neq 0]=Pr[x_i= 1]={\mu}(X)$.
Then we have $Pr[Pr[f(t_i)\neq x_i]\ge lb+\epsilon]=Pr[\mu(X)\ge lb+\epsilon]\le Pr[\mu(X)\ge \mu(O)+\epsilon]\le Pr[|\mu(X)-\mu(O)|\ge \epsilon]$.
According to Lemma~\ref{lemma:b-bound}, we have $Pr[|\mu(X)-\mu(O)|\ge \epsilon]\le  2 \exp\left(-\frac{k\epsilon^2}{2\hat{\sigma}^2 + {2\epsilon}/{3}} \cdot \frac{n - k}{n - 1} \right)< 2\exp(n\ln l)=2l^n$.

For another case that $\mu(O)\ge ub$, we have $f(t_i)=1$. Hence, $Pr[f(t_i)\neq x_i]=Pr[x_i\neq 1]=Pr[x_i=0]=1-{\mu}(X)$.
Then we have $Pr[Pr[f(t_i)\neq x_i]\ge 1-(ub-\epsilon)]=Pr[1-\mu(X)\ge 1-(ub-\epsilon)]=Pr[\mu(X)\le ub-\epsilon]\le Pr[\mu(X)\le \mu(O)-\epsilon]\le Pr[|\mu(X)-\mu(O)|\ge \epsilon]$.
Based on the preceding analysis, we have already established that $Pr[|\mu(X)-\mu(O)|\ge \epsilon]\le 2 l^n$.
\eop

Note that, in the scenario where $lb<|O^{+}|/|O|< ub$, the \textsc{SemanticFilter} continues to repartition the cluster while maintaining the values of $lb$ and $ub$. Therefore, the analysis remains consistent. Therefore, here, we only focus on cases where $\mu(O)=|O^{+}|/|O|$ satisfies either $\mu(O)\le lb$ or $\mu(O)\ge ub$.

\stitle{Analysis of \textsc{SimVote}.}
Next, we analyze the \textsc{SemanticFilter} with \textsc{SimVote}.
We consider a weighted version of Lemma~\ref{lemma:b-bound}.

\begin{corollary}
Let $X = \{x_1, ..., x_n\}$ be a finite population of size $n$, where $x_1,\cdots,x_n$ are binary variables with a mean $\mu$.
Let $c_1,\cdots, c_k$ be sampled without replacement from $X$, and each $c_i$ is associated with a weight $w_i$, where $w_i\ge 0$ and $\sum_{i=1}^k w_i = 1$.
Let $\hat{\mu}$ be the mean of $c_1,\cdots, c_k$, i.e., $\hat{\mu} = \frac{1}{k}\sum_{i=1}^k c_i$; and $\hat{\mu}_w$ be the weighted mean of $c_i$ with weight $w_i$, i.e., $\hat{\mu}_w=\frac{1}{k}\sum_{i=1}^k w_ic_i$.
For any $\epsilon > 0$, we have
    $$Pr [| \hat{\mu}_w - \mu| \ge \epsilon] \le 2\exp(\frac{-3k\epsilon^2}{({6}\hat{\sigma}^2 + {2\epsilon)v}}\cdot\frac{n - k}{n - 1}),$$
    \label{Bernsteinbound:extension}
    where $\hat{\sigma}^2 = \frac{1}{k-1}\sum_{i=1}^k(c_i - \hat{\mu})^2$ is the sample standard deviation, and v is a constant such that $\max_i w_i \le \frac{v}{k}$.
\label{corollary:1}
\end{corollary}

The derivation can be written as $\hat{\mu}_w - \mu = \sum_{i = 1}^k y_i$ with $y_i = w_i(c_i - \mu)$, whose expectation is $0$. Applying Lemma ~\ref{lemma:b-bound} to $y_i$ and bounding the variance, Corollary ~\ref{Bernsteinbound:extension} can be proved. Full derivations and detailed proof steps are provided in appendix ~\ref{app:proof}

Here, given $X = \{x_1, \cdots, x_n\}$ of a cluster $C_j$, where $x_1, \cdots, x_n$ are binary, and $x_i=\mathbb{I}(M(t_i,e))$ for each $(\mathbf{e}_i,t_i)\in C_j$.
Let $g(t_i)\in \{0,1\}$ indicate whether $(\mathbf{e_i},t_i)$ is included in $R^+$ by \textsc{SimVote}, i.e., $g(t_i)=1$ if $\text{score}(t_i)\ge ub$ and $g(t_i)=0$ if $\text{score}(t_i)\le lb$, where $\text{score}(t_i)$ is defined in Eq.~\eqref{eq:simvoate:score}.

\begin{lemma}
Given a cluster $C=\{(\mathbf{e}_1,t_1),\cdots, (\mathbf{e}_n,t_n)\}$ and its subset $O\subseteq C$, which is sampled uniformly from $C$.
Let $X(C) = \{x_1, \cdots, x_n\}$ of $C$, where $x_i=\mathbb{I}(M(t_i,e))$ is binary for each $(\mathbf{e}_i,t_i)\in C$, and $\mu(\text{score}(C\setminus O))$ be the mean value of $\text{score}(t_i)$ (defined in Eq.~\ref{eq:simvoate:score}) for $(\mathbf{e_i},t_i)\in C\setminus O$. We have $\mu(\text{score}(C\setminus O))=\mu(X(O))$.
\label{lemma:simvote}
\end{lemma}

\proofsketch
Since $\mathbb{I}(M(t_j, e))=x_j$, we have
\begin{align}
|C\setminus O|\cdot\mu(\text{score}(C\setminus O))=\sum_{t_i\in C\setminus O}\sum_{t_j \in O} \frac{\text{sim}(t_i, t_j) }{\sum_{t_k \in O} \text{sim}(t_i, t_k)} x_j \nonumber\\
=\sum_{t_j\in O}x_j\cdot \left(\sum_{t_i\in C\setminus O} \frac{\text{sim}(t_i, t_j)}{\sum_{t_k \in O} \text{sim}(t_i, t_k)}\right)\label{eq:cons}.
\end{align}
Since $O$ is sampled uniformly from $C$, we have $\text{r.h.s. of Eq.~\eqref{eq:cons}} =\sum_{t_j\in O}x_j|C\setminus O|$, hence, $\mu(\text{score}(C\setminus O))=\sum_{t_j\in O}x_j=\mu(X(O))$.
\eop

Based on Corollary~\ref{corollary:1} and Lemma~\ref{lemma:simvote}, we confirm that $g(t_i)$ is an accurate estimator of $x_i$ for each $(\mathbf{e_i},t_i)$  when the sampling rate $\xi$ is sufficiently large.

\begin{theorem}
Given a cluster $C=\{(\mathbf{e}_1,t_1),\cdots, (\mathbf{e}_n,t_n)\}$ and parameters $lb$, $ub$, $l$, $v$ and $\epsilon$. Let $X = \{x_1, \cdots, x_n\}$ of $C$, where $x_i=M(t_i,e)$ is binary for each $(\mathbf{e}_i,t_i)\in C$. Assume that $\xi$ satisfies
$$\xi \ge \frac{1}{2} - \sqrt{\frac{1}{4} + \frac{v ln l (6\hat{\sigma}^2 + 2 \epsilon)}{3 \epsilon^2}}$$
Then, the following inequality holds with at least $1-2l^n$ probability when $score(t_i)\le lb$ or $score(t_i)\ge ub$: $$Pr[g(t_i)\neq x_i]\le\max(lb+\epsilon,1-(ub-\epsilon)).$$
\label{theo:simvote}
\end{theorem}
Theorem~\ref{theo:simvote} extends a theoretical guarantee to \textsc{SimVote}. Although each tuple receives a personalized score, Lemma~\ref{lemma:simvote} ensures that these scores concentrate around the same mean as uniform sampling. The condition on $\xi$ in Theorem~\ref{theo:simvote} is slightly more complex than in Theorem~\ref{theo:univote} due to the variance derived from the weighting mechanism. Despite this, the final guarantee takes the same form.

\proofsketch
\sloppy
We have $k=\xi n\ge {n}/{2} - \sqrt{n^2/4 + {n^2 v lnl(6\hat{\sigma}^2 + 2 \epsilon)/{3\epsilon^2}}}$.
For the case $score(t_i)\le lb$, we have $g(t_i)=0$. Hence, $Pr[g(t_i)\neq x_i]=Pr[x_i\neq 0]=Pr[x_i= 1]={\mu}(X)$.
Then $Pr[Pr[g(t_i)\neq x_i]\ge lb+\epsilon]=Pr[\mu(X)\ge lb+\epsilon]\le Pr[\mu(X)\ge score(t_i)+\epsilon]\le Pr[|\mu(X)-score(t_i)|\ge \epsilon]$.
According to Lemma~\ref{lemma:simvote}, we have $\mu(score(t_i))=\mu(X)$ and based on Corollary~\ref{corollary:1}, we have $Pr[|\mu(X)-score(t_i)|\ge \epsilon]\le  2exp(\frac{-3k\epsilon^2}{({6}\hat{\sigma}^2 + {2\epsilon})v}\cdot\frac{n - k}{n - 1})
< 2\exp(n\ln l)=2l^n$.
For another case that $score(t_i)\ge ub$, we have $g(t_i)=1$. Hence, $Pr[g(t_i)\neq x_i]=Pr[x_i\neq 1]=Pr[x_i=0]=1-{\mu}(X)$.
Then we have $Pr[Pr[g(t_i)\neq x_i]\ge 1-(ub-\epsilon)]=Pr[1-\mu(X)\ge 1-(ub-\epsilon)]=Pr[\mu(X)\le ub-\epsilon]\le Pr[\mu(X)\le score(t_i)-\epsilon]\le Pr[|\mu(X)-score(t_i)|\ge \epsilon]$.
Based on the preceding analysis, we have already established that $Pr[|\mu(X)-score(t_i)|\ge \epsilon]\le 2 l^n$.

\section{Experimental Study}
\label{sec:exp}

In this section, we present the experimental setup in \cref{sec:exp:setup} and evaluate our approaches comprehensively in the following facets: 
\ding{202} Compare the baselines on the semantic filter and present the results of overall effectiveness and efficiency in \cref{sec:exp-filter}.
\ding{203} Analyze the effects of hyper-parameters on CSV in \cref{sec:exp-parameter}.
\ding{204} Investigate the efficacy of re-clustering in CSV in \cref{sec:exp:reclustering}.
\ding{205} Explore the disparity between practical and theoretical results in \cref{sec:exp-theo}.
\ding{206} Evaluate the generalizability of CSV on diverse synthetic queries in \cref{sec:exp:synthetic}. 
\ding{207} Study the impacts of employing different embedding models and LLMs in appendix~\ref{app:exp}. 

\subsection{Experimental Setup}
\label{sec:exp:setup}

\begin{table*}[t]
    \caption{Profile of Test Semantic Queries}
    \vspace{-3mm}
    \centering
    \footnotesize
    \label{tab:queryset}
    \setlength{\tabcolsep}{2pt}
    \begin{tabular}{c c|c l c}
    \toprule
     \textbf{Dataset} & \textbf{Query} & \textbf{Query Type}& \textbf{Predicate} & \textbf{Ground-truth}\\
    \midrule
   \multirow{3}{*}{IMDB-Review} & RV-Q1     & semantic filter & The \{\kwnospace{review}\} is positive. & available  \\
        & RV-Q2    & semantic filter & The \{\kwnospace{review}\} explicitly recommends that others watch the movie. & GPT-4o \\
        & RV-Q3    &  semantic filter & The \{\kwnospace{review}\} focuses on factual descriptions like technical details. & GPT-4o \\
    \multirow{3}{*}{Codebase} & CB-Q1  & semantic filter & The \{\kwnospace{AboutMe}\} contains a link to social media. & GPT-4o \\
        & CB-Q2    & semantic filter & {\makecell[l]{The \{\kwnospace{AboutMe}\} text suggests that the user has either experience in \\computer science or expresses an interest in the field of computer science.}} & GPT-4o\\
        & CB-Q3  & semantic filter &  The \{\kwnospace{AboutMe}\} focuses on factual identification details. & GPT-4o \\
    TC &   TC     & semantic filter  & The \{\kwnospace{comment}\} contains hate speech or offensive language. & available \\
    Airdialogue &AD-Q1/Q2/Q3/Q4 & semantic filter & The \{\kwnospace{dialogue}\} has outcome \kwnospace{book}/\kwnospace{cancel}/\kwnospace{no\_flight}/\kwnospace{no\_reservation}. & available \\
    Fever & Fever    & multi-column filter & The \{\kwnospace{claim}\} is supported by the \{\kwnospace{evidence}\}. & available \\
    \bottomrule
    \end{tabular}
\end{table*}

\stitle{Datasets.} 
Following~\cite{patel2024lotus,liu2023,biswal2024text2sql, DBLP:journals/pacmmod/ZeighamiSP25} we evaluate five datasets for semantic filter.
\textbf{\ding{172} IMDB-Review:} The IMDB-Review~\cite{IMDB-review} dataset is a binary sentiment classification benchmark derived from the movie review domain.
It comprises 50,000 labeled reviews, which are full-length user-written reviews from Internet Movie Database (IMDB), each expressing either a positive or negative opinion about a film. The dataset is evenly split between the two sentiment classes.
\textbf{\ding{173} Codebase:} The Codebase dataset~\cite{bird} consists of a structured table named \kw{user} with 9,378 rows, which includes a column titled \kw{AboutMe} containing long-form self-introduction texts. 
\textbf{\ding{174} Airdialogue:} The Airdialogue dataset~\cite{airdialogue} is a classification task derived from the airline ticketing domain. It comprises 402,035 unique dialogs, each annotated with one of five possible categories: \kw{book} $(51.40\%)$, \kw{cancel} $(1.46\%)$, \kw{no\_flight} $(23.08\%)$, \kw{no\_reservation} $(23.89\%)$, and \kw{change} $(0.17\%)$.
\textbf{\ding{175} TC:} The Twitter Hate Speech and Offensive Language (TC) dataset~\cite{DBLP:conf/icwsm/tc} is a classification benchmark derived from the social media domain. It comprises 24,783 labeled tweets, each annotated into one of three categories: \kw{hate speech}, \kw{offensive language}, or \kw{neither}. The dataset was collected from Twitter using hate speech-related keywords and subsequently annotated via crowdsourcing. 
\textbf{\ding{176} Fever:} The Fact Extraction and Verification (Fever) dataset~\cite{DBLP:conf/naacl/fever} is a large-scale benchmark for claim verification in the fact-checking domain, constructed by altering sentences extracted from Wikipedia. In this work, we use a cleaned version from which  incomplete entries have been removed to ensure consistency and reliability.

\stitle{Queries.}
Table~\ref{tab:queryset} summarizes the 14 test queries used in our experiments. 
For RV-Q2/Q3 and CB-Q1/Q2/Q3, ground-truth labels are not publicly available; we therefore use GPT-4o as a proxy to obtain these labels via direct LLM calls. 
Among the 12 queries, Fever is a multi-column semantic filter that references two columns, the remaining queries are single-column predicates.
The task on Airdialogue is a multi-class classification problem that determines the category of a dialog. We decompose the task into four binary semantic filter queries, AD-Q1/Q2/Q3/Q4, where each query is a predicate for one of the four classes: \kwnospace{book}/\kwnospace{cancel}/\kwnospace{no\_flight}/\kwnospace{no\_reservation}.
In addition to predicates shown in Table~\ref{tab:queryset}, we prepend a following instruction to the prompt for all Airdialogue queries: `The following is a dialog between an airline ticketing agent and a customer. The outcome of the dialog will be one of the following 5 categories. [book]: the agent has booked a flight for the customer (not including the flight change). [cancel]: the agent canceled the existing valid reservation for the customer. [change]: the agent changed the existing flight reservation of the customer and successfully found a new one. [no\_reservation]: the customer wants to change or cancel the flight, but there is no valid reservation under this customer. [no\_flight]: the customer aims to book a flight from departure to destination but finds no flights between departure and destination.'

\begin{table*}[t]
\centering
\footnotesize
\caption{The Accuracy and F1 Score of Semantic Filter Queries}
    \vspace{-3mm}
\begin{tabular}{c | lc c c c cccccc c c}
\toprule
 & \textbf{Method} & \textbf{RV-Q1} & \textbf{RV-Q2} & \textbf{RV-Q3} & \textbf{CB-Q1} & \textbf{CB-Q2} & \textbf{CB-Q3} & \textbf{AD-Q1} & \textbf{AD-Q2} & \textbf{AD-Q3} & \textbf{AD-Q4} & \textbf{TC} & \textbf{Fever} \\
\midrule
\multirow{5}{*}{\textbf{Acc.}} & Reference & 0.9480 & 0.8229  & 0.9097  & 0.9731 & 0.7782  & 0.7033 & 0.8972 & 0.9139 & 0.8112 & 0.8737 & 0.7203  &  0.8258\\
 & \lotus  & 0.9480  & 0.8679  & 0.9948  & 0.9595 & 0.6853 & 0.6996 & 0.8938 & 0.9066 & 0.9192 & 0.8690  & 0.5196  &  0.7396\\
& \BARGAIN   &   0.9186   & 0.7881  & 0.9948   &   0.9681      & 0.5308 & 0.6493  & 0.7423   & 0.8972 & 0.3885 & 0.9563 & 0.7872 & 0.5177\\
& \CSVU    & 0.9307  & 0.7723 & 0.9527 & 0.9673 & 0.7810 & 0.7090 & 0.8948 & 0.9193 & 0.8723 & 0.8975  & 0.7477  & 0.7526 \\
& \CSVS  & 0.9305  & 0.7718  & 0.9398  & 0.9671 & 0.7813 & 0.7376  & 0.8968 & 0.9193 & 0.8722 & 0.8977  & 0.7799  & 0.7524 \\
\midrule
\multirow{5}{*}{\textbf{F1}} & Reference & 0.9489 & 0.7633 & 0.0880 & 0.6344 & 0.7557 & 0.6374 & 0.9056 & 0.8165 & 0.7002 & 0.1841  & 0.8060  & 0.8947\\
& \lotus & 0.9489 & 0.7113 &0.0972  & 0.4947 & 0.6115 & 0.3107 & 0.9029 & 0.8048 & 0.8393 & 0.1788  & 0.6007  & 0.7955 \\
& \BARGAIN   & 0.9140   & 0.5616 & 0.0296    & 0.5954  & 0.1979 & 0.0947 & 0.6878 & 0.7722 & 0.3429  &  0.3697 & 0.8589  & 0.5831 \\
& \CSVU & 0.9336 & 0.7097 & 0.0880  & 0.0623 & 0.7601 & 0.6340 & 0.9039 & 0.8249 & 0.7800 & 0.2179   & 0.8348 & 0.8588 \\
& \CSVS & 0.9335 & 0.7092 & 0.0906 & 0.0679 & 0.7555 & 0.6589  & 0.9052 & 0.8249 & 0.7799 & 0.2175  & 0.8550 & 0.8586 \\
\bottomrule
\end{tabular}

\label{t:filter_overall}
\end{table*}

\begin{figure*}
    \begin{tabular}[t]{c}
        \centering
        \includegraphics[width=0.6\columnwidth]{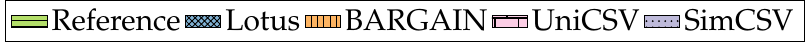}
	\end{tabular}
    \begin{tabular}[h]{c}
        \subfigure[Query Time]{
				\includegraphics[width=0.67\columnwidth]{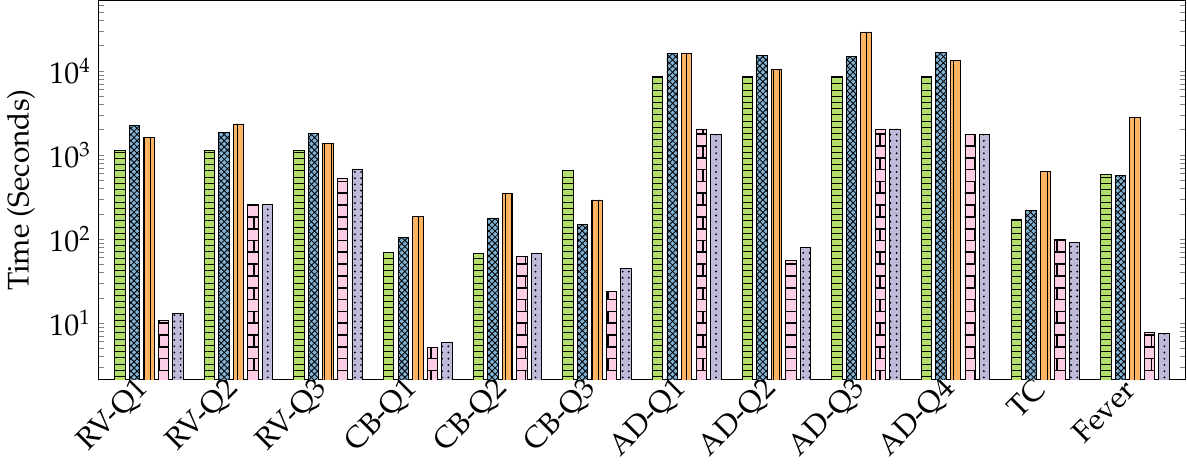}
			\label{fig:T_heatmap}
		} 
        \subfigure[Total \# LLM Calls] {
			\includegraphics[width=0.67\columnwidth]{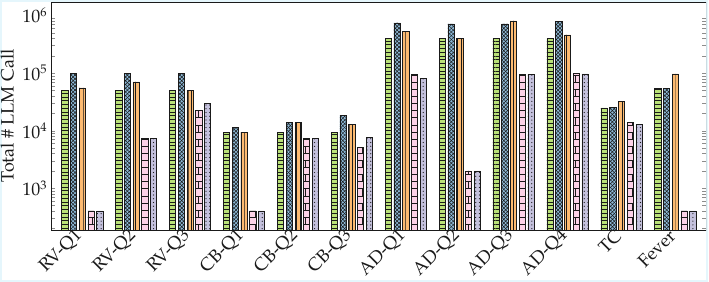}
			\label{fig:F_heatmap}
		} 
        \subfigure[Total \# Token] {
			\includegraphics[width=0.67\columnwidth]{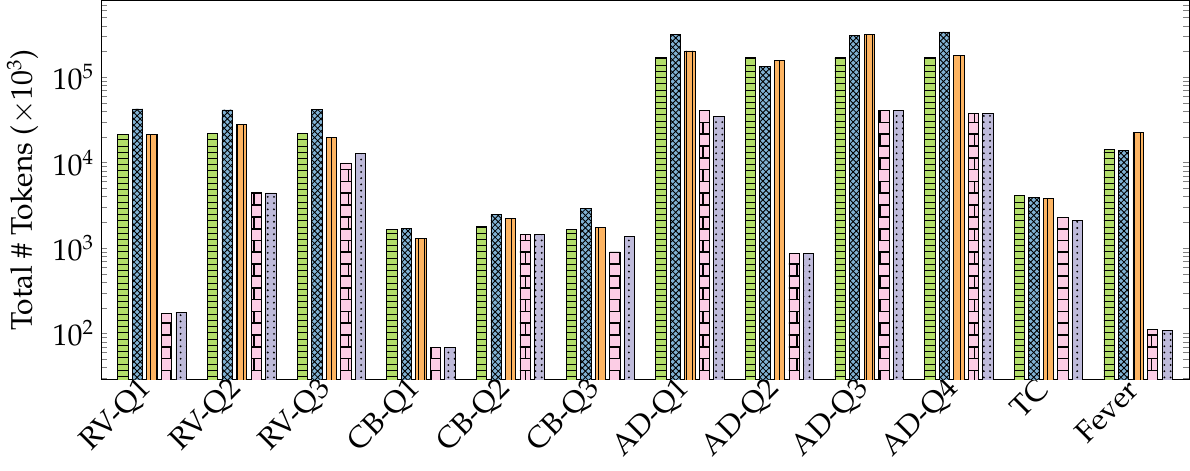}
			\label{fig:F_heatmap}
		}
	\end{tabular}
    
    \vspace{-5mm}
	\caption{Comparisons of Query Time, \# LLM Calls and Token Usage for Semantic Filter}
        \label{fig:casestudy}
    \vspace{-3mm}
\end{figure*}

\stitle{Algorithms.}
We compare our approaches with the following existing methods.
\textbf{\ding{172} \basefilter}: Following the definition in Eq.~\eqref{eq:sem_filter}, we sequentially invoke the LLM for each tuple and retrieve the results to establish a baseline for the semantic filter.
\textbf{\ding{173} \lotus}: \lotus introduces state-of-the-art approaches for the semantic filter~\cite{patel2024lotus}. We adopt the approach as a baseline for the respective tasks.
\textbf{\ding{174} \BARGAIN}: \BARGAIN~\cite{DBLP:journals/pacmmod/ZeighamiSP25} adopts a similar model cascade framework as~\lotus while providing strong theoretical guarantees on  Precision, Recall, and Accuracy. We configure \BARGAIN with an accuracy target of 0.85 and a tolerance of 0.05.
For both \lotus and \BARGAIN we use the default hyper-parameters provided in their respective source code repositories. Both methods use LlaMA-3.2-3B for initial filtering and LlaMA-3.1-8B as the larger model for secondary processing.

For our approaches,
\textbf{\ding{176} \CSVU} and \textbf{\ding{177} \CSVS} are two algorithms for the semantic filter (Algorithm~\ref{alg: semantic_filter}), instantiated with \textsc{UniVote} (Algorithm~\ref{alg:univote}) and \textsc{SimVote} (Algorithm~\ref{alg:simvote}), respectively.
By default, we employ LLaMA-3.1-8B~\cite{touvron2023llama} as the backbone LLM and E5-Large~\cite{wang2022e5} as the embedding model. 
When input text exceeds the context length supported by the embedding model, it is segmented into smaller chunks, each embedded independently. The final embedding is computed as the mean of all segment embeddings to ensure a consistent representation across variable-length inputs.
For clustering, we employ the {K-means algorithm from~\cite{PedregosaVGMTGBPWDVPCBPD11}} and set the default number of clusters to 4.
To prevent excessive re-clustering during the filtering process, the maximum number of re-cluster iterations is capped at 3. Once a voting failure occurs in the remaining clusters, we resort to sequential LLM invocation to ensure bounded errors. 
To better accommodate the diversity of predicates, we further introduce BM25 score~\cite{JonesWR00} that captures lexical similarity, as a distance metric for clustering, used jointly with the embedding Euclidean distance. A hyper-parameter $\lambda$ is introduced to balance the relative contributions between distances, i.e., $\lambda$ for L2 distance and $(1 - \lambda)$ for BM25.
We set $\lambda = 0.4$ for CB-Q1 and TC, and $\lambda = 1.0$ for the remaining queries.
Other hyper-parameters are as follows, $lb = 0.15$, $\xi=0.005$, $\text{min\_sample}=101$.

\stitle{Evaluation Metrics.}
We evaluate different methods by Accuracy, Recall, F1 score; and efficiency by execution time, \# LLM calls, and token consumption.
This multifaceted evaluation enables a balanced analysis of both predictive effectiveness and computational efficiency.
For efficiency, {execution time} measures the end-to-end wall-clock time required to complete the semantic operator.
{\# LLM call} refers to the total number of invocations of the LLMs. For consistency, we count calls based on the number of prompts issued, irrespective of whether multiple tuples are batched within a single prompt.
{Token consumption} captures the total number of tokens processed, including both input (prompt) and output tokens.

\stitle{Implementation Details}
The experiments are conducted locally on a single 80GB NVIDIA A100 GPU, with vLLM~\cite{kwon2023efficientvllm} serving as the inference engine. For all approaches, batch prompt processing is handled using the \textit{batch\_completion} method from LiteLLM~\cite{litellm2023}, with a temperature of 0.7 and a maximum output token limit of 32.

\begin{table*}[t]
\centering
\footnotesize
\caption{Hyper-parameter Effects on the Performance of \CSVU and \CSVS}
    \vspace{-3mm}
\begin{tabular}{lr|cccccc|cccccc}
\toprule
\multirow{3}{*}{\bf Parameter} & \multirow{3}{*}{\bf Value}
& \multicolumn{2}{c}{\textbf{RV-Q1}} & \multicolumn{2}{c}{\textbf{CB-Q2}} & \multicolumn{2}{c}{\textbf{AD-Q2}} & \multicolumn{2}{c}{\textbf{RV}} & \multicolumn{2}{c}{\textbf{CB-Q2}} & \multicolumn{2}{c}{\textbf{AD-Q2}} \\
 & & \textbf{Acc.} & \textbf{F1} & \textbf{Acc.} & \textbf{F1} & \textbf{Acc.} & \textbf{F1}& \textbf{Acc.} & \textbf{F1} & \textbf{Acc.} & \textbf{F1} & \textbf{Acc.} & \textbf{F1} \\
\cmidrule{3-14}
& & \multicolumn{6}{c}{\textbf{\CSVU}} & \multicolumn{6}{|c}{\textbf{\CSVS}} \\
\midrule
\multirow{4}{*}{\# clusters} 
 & 2  & 0.9079 & 0.9138 & 0.7780 & 0.7492 & 0.9168 & 0.8210 & 0.9082 & 0.9140 & 0.7739 & 0.7450 & 0.9168 & 0.8209 \\
 & 4  & 0.9307 & 0.9336 & 0.7852 & 0.7601 & 0.9192 & 0.8248 & 0.9305  & 0.9335 & 0.7803 & 0.7539 & 0.9193 & 0.8249 \\
 & 8  & 0.9367 & 0.9367 & 0.7829 & 0.7597 & 0.9202 & 0.8264 & 0.9388  & 0.9384 & 0.7774 & 0.7534 & 0.9202 & 0.8264 \\
 & 16 & 0.9424 & 0.9438 & 0.7815 & 0.7564 & 0.9200 & 0.8262 & 0.9379 & 0.9396 & 0.7788 & 0.7546 & 0.9199 & 0.8260 \\
\midrule
\multirow{5}{*}{\shortstack{\\sampling\\rate $\xi$ ($\text{\textperthousand}$)}} 
 & 5 & 0.9307 & 0.9336 & 0.7852 & 0.7601 & 0.9192 & 0.8248 & 0.9305 & 0.9335 & 0.7803 & 0.7539 & 0.9193 & 0.8249 \\
 & 10 & 0.9304 & 0.9334 & 0.7794 & 0.7518 & 0.9192 & 0.8248 & 0.9459 & 0.9470 & 0.7756 & 0.7493 & 0.9193 & 0.8249 \\
 & 15 & 0.9304 & 0.9334 & 0.7802 & 0.7537 & 0.9193 & 0.8250 & 0.9439 & 0.9453 & 0.7828 & 0.7568 & 0.9192 & 0.8248 \\
 & 20 & 0.9308 & 0.9337 & 0.7803 & 0.7536 & 0.9191 & 0.8246& 0.9306 & 0.9335 & 0.7827 & 0.7565 & 0.9191 & 0.8248 \\
 & 25 & 0.9307  & 0.9336 & 0.7788 & 0.7513 & 0.9192 & 0.8248 & 0.9307  & 0.9337 & 0.7861 & 0.7603 & 0.9191 & 0.8246 \\
\midrule
\multirow{4}{*}{\shortstack{lower\\bound}}
 & 0.10  & 0.9302 & 0.9333 & 0.7803 & 0.7547 & 0.9193 & 0.8249 & 0.9312 & 0.9341 & 0.7814 & 0.7551 & 0.9193 & 0.8249 \\
 & 0.15 & 0.9307 & 0.9336 & 0.7852 & 0.7601 & 0.9192 & 0.8248 & 0.9305  & 0.9335 & 0.7803 & 0.7539 & 0.9193 & 0.8249 \\
 & 0.20  & 0.9304 & 0.9334 & 0.7852 & 0.7595 & 0.9193 & 0.8249 & 0.9304 & 0.9334 & 0.7839 & 0.7575 & 0.9193 & 0.8248 \\
 & 0.50  & 0.9304 & 0.9334 & 0.7959 & 0.7956 & 0.9193 & 0.8249 & 0.9304  & 0.9334 & 0.7953 & 0.7950 & 0.9193 & 0.8345 \\
\bottomrule
\end{tabular}
\label{t:hyperparameters}
\end{table*}
\subsection{Overall Evaluation of Semantic Filter}
\label{sec:exp-filter}

In this section, we compare our two approaches, \CSVU and \CSVS, against the baselines on both efficiency and effectiveness. 

\stitle{Efficiency.} 
Fig.~\ref{fig:casestudy} illustrates LLM calls, execution time, and token consumption of different approaches. \CSVU and \CSVS achieve substantial efficiency gains by reducing LLM calls  by 1.28-$200\times$ compared to \basefilter, 1.81-$355\times$ compared to \lotus, and 1.68-$200\times$ compared to \BARGAIN. Since LLM calls dominate both runtime and token usage, these reductions translate directly into one to three orders of magnitude improvements across all datasets. For instance, on RV-Q1, \CSVU and \CSVS require only 404 calls, completing in under 13 seconds with approximately 170k tokens. In contrast, the baselines require tens of thousands of calls, resulting in runtimes exceeding 1,000 seconds and token usage exceeding 20 million.
The higher cost incurred by \lotus and \BARGAIN stems from their two-stage model-cascaded filtering. In many cases, the smaller 3B model fails to learn a threshold with sufficient precision, leading to more than half of the data being routed again to the larger 8B model.
In the worst cases, i.e., RV-Q1 for \lotus and AD-Q3 for \BARGAIN, the smaller model fails to filter any tuple, thereby increasing cost across all three metrics.

\stitle{Effectiveness.}
Table~\ref{t:filter_overall} reports the Accuracy and F1 scores of all approaches across the 12 semantic filter queries. Overall, \CSVU and \CSVS demonstrate comparable performance with \basefilter, and outperform \lotus and \BARGAIN, except CB-Q1.
In contrast, \BARGAIN exhibits highly unstable quality compared to \basefilter: 
while it achieves strong F1 scores and Accuracy on AD-Q4, it suffers from substantial degradation on CB-Q2 and AD-Q3. We speculate that this instability is because it uses the logits of LLM as the confidence signals of the predication, which is unreliable due to the inherent overconfidence of neural networks~\cite{DBLP:conf/icml/GalG16, DBLP:conf/naacl/GengCWKNG24}.
One notable exception is the lower F1 score of our approaches on CB-Q1. This query explicitly mentions the entity `social link', thereby possesses an extremely low selectivity of 0.033.
The rare positives make Recall unstable with sampling failures under a default setting of $lb=0.15$ (many samples contain few or no positives) and subsequent voting overconfidence bypasses the re-clustering. When we change $lb$ to 0.01, the F1 score reaches 0.5825, but around half of the tuples are processed by linear LLM invocation. 

The distinction between \CSVU and \CSVS is marginal in terms of Accuracy and F1 scores on well-clustered datasets, indicating their comparable effectiveness. In less clustered cases like CB-Q2 and AD-Q1, \CSVS exhibits a slight edge, showcasing more resilience to noisy or ambiguous cluster boundaries. This robustness enhances filter efficiency under challenging conditions.
Moreover, \CSVU and \CSVS exhibit strong stability, consistently decreasing time and token consumption while maintaining effectiveness across datasets, queries, and varying sizes, from thousands (e.g., CB) to hundreds of thousands (e.g., AD).

\subsection{Hyper-parameters Analysis}
\label{sec:exp-parameter}
We analyze three hyper-parameters, number of clusters, sample ratio, and lower bound, that affect the effectiveness and efficiency of \CSVU and \CSVS as follows.

\stitle{(1) Number of Clusters.}
As shown in Table~\ref{t:hyperparameters}, enlarging the cluster size enhances practical performance, although it does not impact the theoretical error bound.
Accuracy and F1 score notably improve when \#clusters range from 2 to 4, with gains diminishing thereafter.

For efficiency, depicted in Fig.~\ref{fig:para_filter}, the number of clusters does not affect the total LLM calls, such as on RV-Q1 and AD-Q2, as the sample ratio dictates queries per cluster, maintaining a stable sampling volume. In smaller datasets like CB-Q2, the number of samples per cluster is primarily regulated by \textit{min\_sample}, resulting in a linear increase in total LLM calls with more clusters.
Furthermore, in AD-Q2, setting the cluster count to 2 triggers re-clustering due to insufficient semantic distinction, leading to increased LLM calls, runtime, and token usage.

\begin{figure}[t]
    \begin{tabular}[t]{c}
        \centering
        \includegraphics[width=0.6\columnwidth]{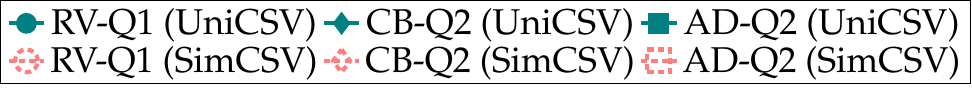}
	\end{tabular}
    \begin{tabular}[t]{c}
        \subfigure[\# Clusters] {
				\includegraphics[width=0.32\columnwidth]{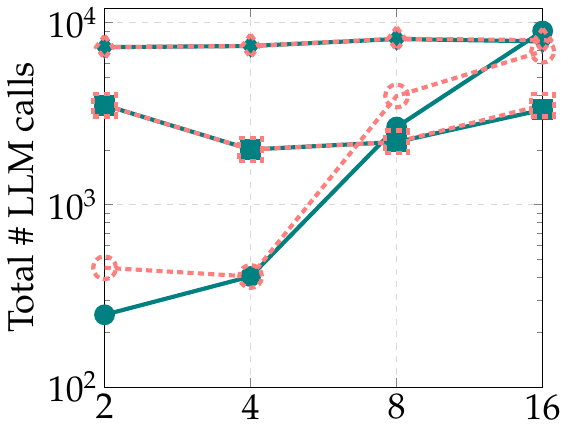}
			\label{fig:T_heatmap}
		} 
        \hspace{-0.3cm}
        \subfigure[Sampling Rate ($\text{\textperthousand}$)] {
			\includegraphics[width=0.32\columnwidth]{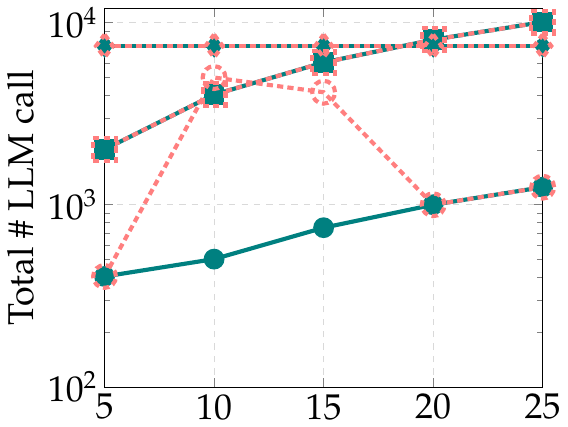}
			\label{fig:F_heatmap}
		}
        \hspace{-0.3cm}
        \subfigure[lower bound $lb$] {
			\includegraphics[width=0.32\columnwidth]{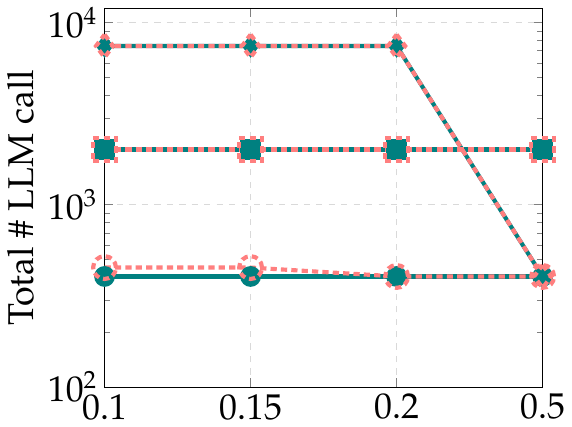}
			\label{fig:F_heatmap}
		}
	\end{tabular}
    \begin{tabular}[t]{c}
        \subfigure[\# Clusters] {
				\includegraphics[width=0.32\columnwidth]{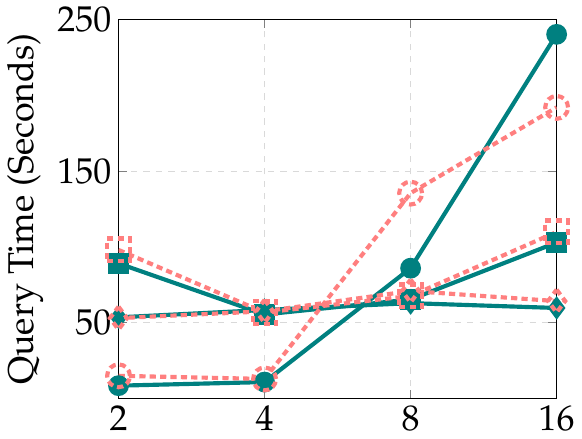}
			\label{fig:T_heatmap}
		} 
        \hspace{-0.3cm}
        \subfigure[Sampling Rate ($\text{\textperthousand}$)] {
			\includegraphics[width=0.32\columnwidth]{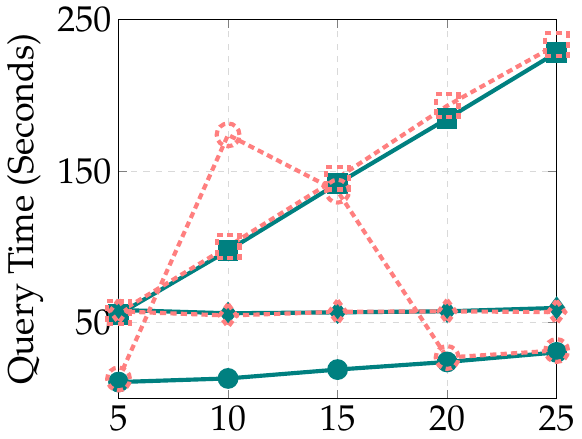}
			\label{fig:F_heatmap}
		}
        \hspace{-0.3cm}
        \subfigure[lower bound $lb$] {
			\includegraphics[width=0.32\columnwidth]{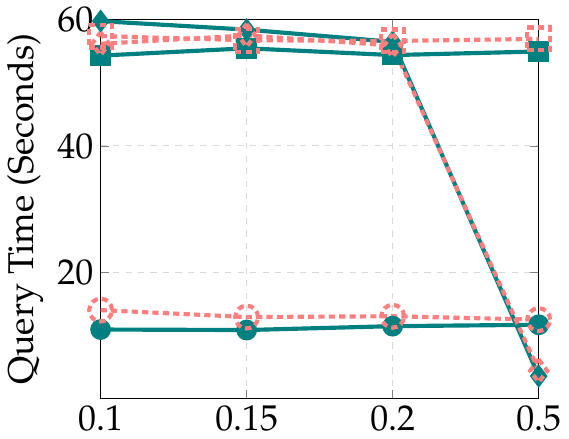}
			\label{fig:F_heatmap}
		}
	\end{tabular}
     \begin{tabular}[t]{c}
        \subfigure[\# Clusters] {
				\includegraphics[width=0.32\columnwidth]{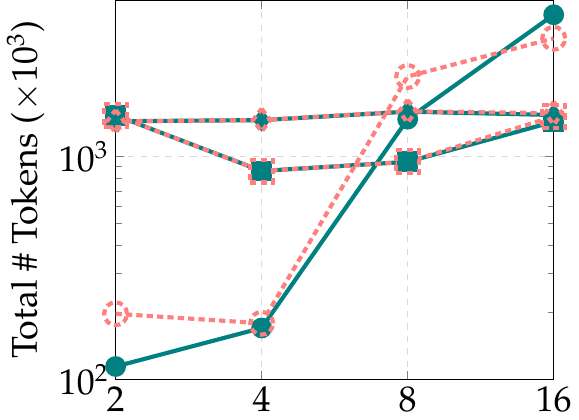}
			\label{fig:T_heatmap}
		} 
        \hspace{-0.3cm}
        \subfigure[Sampling Rate ($\text{\textperthousand}$)] {
			\includegraphics[width=0.32\columnwidth]{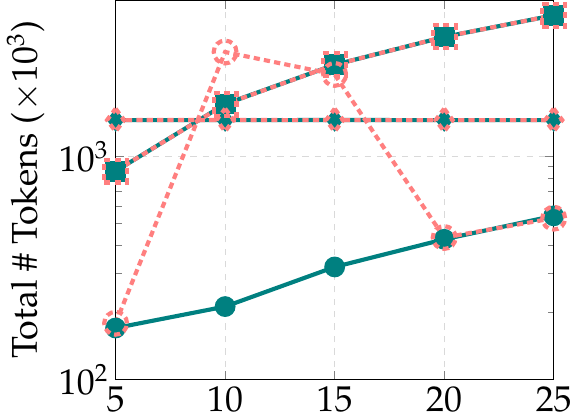}
			\label{fig:F_heatmap}
		}
        \hspace{-0.3cm}
        \subfigure[lower bound $lb$] {
			\includegraphics[width=0.32\columnwidth]{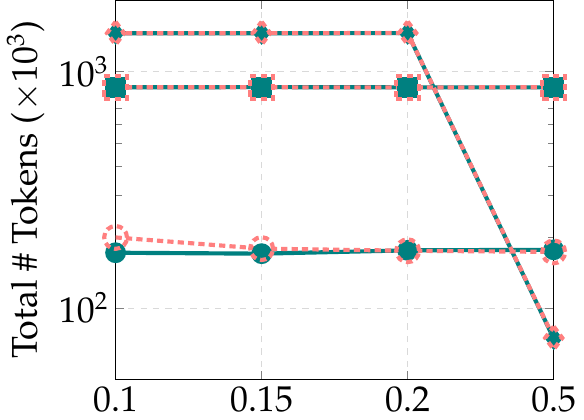}
			\label{fig:F_heatmap}
		}
	\end{tabular}
    \vspace{-3mm}
	\caption{Hyper-parameter Effects on Query Time, \# LLM Calls and Token Usage for Semantic Filter}
        \label{fig:para_filter}
\end{figure}

\stitle{(2) Sample Ratio.}
As indicated in Table~\ref{t:hyperparameters}, the sample ratio has almost no impact on the Accuracy and F1 score of \CSVU and \CSVS, suggesting that a small sampling ratio suffices. We will delve into this further in the subsequent experiment in \cref{sec:exp-theo}.

In Fig.~\ref{fig:para_filter}, an increase in the sample ratio generally leads to a linear growth in LLM calls across most datasets. While on CB-Q2, LLM calls remain consistent despite the increase in sample ratio. This stability stems from the relatively small size of the dataset, where the sampling number is \textit{min\_sample} in most clusters.
Notably, on RV-Q1 with \CSVS, a significant rise in LLM calls occurs when the sampling rate reaches $10\text{\textperthousand}$. This surge is attributed to partial fallbacks to linear processing, which occurs even within individual clusters. Unlike \CSVU, \CSVS employs a more adaptable voting strategy that may trigger more LLM queries when consensus is weak, even within internally coherent clusters. This illustrates the sensitivity of \CSVS to intra-cluster disagreement, particularly when the sampling density exceeds a threshold that reveals ambiguous cases.

\stitle{(3) Lower Bound.}
By default, we set $ub=1-lb$. The number of LLM calls, Accuracy, and F1 score remain steady when $lb$ is no greater than $0.2$. Once $lb$ hits 0.5 - enforcing voting irrespective of confidence - the LLM call count dramatically decreases, e.g., on CB-Q2, as this $lb$ value averts linear processing and utilizes clustering. Furthermore, on the CB-Q2 query, setting $lb$ to 0.5, which mandates voting solely based on consensus regardless of confidence, yields final accuracy and F1 scores surpassing those of the baselines. This implies that in situations where direct LLM inference is unreliable, i.e., on CB-Q2 queries, leveraging cluster-based consensus can enhance filtering performance beyond what traditional methods can achieve.
These results illustrate that Accuracy and F1 scores stay relatively constant across a broad range of values for these parameters, showcasing the algorithm's robustness to moderate fluctuations in sampling and voting thresholds.

\subsection{Impact of Re-clustering}
\label{sec:exp:reclustering}

We assess the impact of the recursive re-clustering mechanism through an ablation study, comparing both effectiveness metrics (Accuracy and F1) and efficiency metrics (\# LLM calls and re-clustering time) with and without re-clustering. We use the default $lb = 0.15$ as re-clustering is enabled.
When re-clustering is disabled, we fix $lb = ub = 0.5$. This configuration forces the algorithm to commit to a voting decision regardless of confidence, bypassing the recursive re-clustering mechanism.
The results are summarized in Table~\ref{tab:recluster}. 
We observe two distinct behavioral patterns across the evaluated queries. For queries RV-Q1, CB-Q1, AD-Q2, re-clustering is not triggered, yielding nearly identical performance between the two configurations. In RV-Q1 and AD-Q2, the initial clustering produces sufficiently pure clusters where the voting confidence consistently exceeds the threshold bounds, resulting in zero overhead in both additional LLM calls and runtime. In CB-Q1, the sampling failure from the low selectivity leads to over-confident voting results, and the re-clustering is not triggered. 
When using $lb=0.01$ for CB-Q1, re-clustering will be triggered and the accuracy is improved.
In contrast, queries RV-Q3, CB-Q2, and CB-Q3 exhibit substantial degradation when re-clustering is disabled. Accuracy drops by up to 9.7\% and F1 scores decline by more than 12\% in CB-Q3. Furthermore, the number of LLM calls increases significantly when re-clustering is enabled, indicating that a substantial portion of tuples in these cases require finer-grained cluster refinement. This underscores the mechanism's critical role in correcting low-confidence assignments and sustaining prediction quality for queries where the initial clustering fails to separate semantic categories cleanly.
Importantly, despite the increase in LLM calls, the computational overhead of re-clustering remains modest. Even in the most demanding scenarios, re-clustering accounts for less than 3.3\% of total runtime (CB-Q3), demonstrating that the iterative refinement process adds negligible latency relative to the dominant cost of LLM inference.

\begin{table}[t]
    \centering
    \scriptsize
    \caption{Ablation Study of Re-clustering}
    \vspace{-1mm}
     \label{tab:recluster}
        \setlength{\tabcolsep}{2pt}
        \begin{tabular}{c | l cccccc}
        \toprule
          & Method & \textbf{RV-Q1} & \textbf{RV-Q3} & \textbf{CB-Q1} & \textbf{CB-Q2} & \textbf{CB-Q3} & \textbf{AD-Q2}  \\
        \midrule
        \multirow{4}{*}{\textbf{Acc.}} & \CSVU (w/o RC) & 0.9229 & 0.9939  & 0.9677 & 0.7475 & 0.6403 & 0.9243  \\
            &   \CSVS (w/o RC) & 0.9230 & 0.9939 & 0.9677 & 0.7469 & 0.6403 & 0.9240  \\
            &   \CSVU (w/ RC) & 0.9301 & 0.9527 &0.9673 & 0.7813  & 0.7090 & 0.9193 \\
            &   \CSVS (w/ RC) & 0.9305 & 0.9398 & 0.9671 & 0.7813 &  0.7376 & 0.9193  \\
        \midrule
         \multirow{4}{*}{\textbf{F1}} & \CSVU (w/o RC) & 0.9272 & 0 & 0.0619 & 0.7010 & 0.5302 & 0.8347 \\
            &   \CSVS (w/o RC) & 0.9273 & 0  & 0.0619 & 0.7003 & 0.5302 & 0.8345   \\
            &   \CSVU (w/ RC) & 0.9336 & 0.0880 & 0.0623 & 0.7601 & 0.6340 & 0.8249  \\
            &   \CSVS (w/ RC) & 0.9335 & 0.0906 & 0.0679 & 0.7555 & 0.6589 & 0.8249  \\
        \midrule
        \multirow{4}{*}{\textbf{\# Calls}} & \CSVU (w/o RC) & 404 & 404 & 404 & 404 & 404 & 2008 \\
            &   \CSVS (w/o RC) & 404 & 404 & 404 & 404 & 404 & 2008  \\
            &   \CSVU (w/ RC) & 404 & 22760 & 404 & 7423 & 5196  & 2008  \\
            &   \CSVS (w/ RC) &  404 & 29683 & 404 & 7378 & 7586 & 2008  \\
         \midrule
       
        \multirow{2}{*}{\textbf{RC Time (\%)}} & \CSVU (w RC) & 0 & 0.44 & 0 & 1.47 & 3.27 & 0  \\
            &  \CSVS (w RC) & 0 & 0.29 & 0 & 1.37  & 2.05  & 0 \\
        \bottomrule    
        \end{tabular}
\end{table}

\begin{table}[t]
    \centering
\footnotesize
    \caption{Gap between Theory and Practice}
    \label{tab:theoretical}
    \vspace{-1mm}
    \begin{tabular}{c |cc r | cc r}
    \toprule
    \multirow{2}{*}{$\mathbf{\epsilon}$}& \multicolumn{3}{c}{\CSVU}& \multicolumn{3}{c}{\CSVS} \\
     & \textbf{Est.}& \textbf{Err.} & \textbf{$\xi (\text{\textperthousand}) $} & \textbf{Est.} & \textbf{Err.} & \textbf{$\xi (\text{\textperthousand})$}   \\
    \midrule
    $0.10$ & $0.9991$ & $0.0628$ & $13.0$ & $0.9980$ & $0.0622$ & $26.4$\\
    $0.15$ & $0.9997$ & $0.0633$ & $6.4$ & $0.9997$ & $0.0635$ & $12.7$\\
    $0.20$ & $0.9997$ & $0.0636$ & $3.9$ & $0.9997$ &$0.0637$ & $7.8$\\
    $0.25$ & $0.9999$ & $0.0637$ & $2.7$ & $0.9999$ &$0.0635$ & $5.4$\\
    $0.30$ & $0.9999$ & $0.0635$ & $2.0$ & $0.9999$ & $0.0637$ & $4.0$\\
    \bottomrule
    \end{tabular}
    
\end{table}

\subsection{Theoretical Analysis}
\label{sec:exp-theo}
To investigate the gap between theory and practice, we analyze a representative cluster on RV, which contains 14,608 embeddings, and the confidence interval within this cluster is measured at 0.9942, using voting algorithms \CSVU and \CSVS. Specifically, we vary the theoretical error bound $\epsilon$ to compute the required sample ratio $\xi$ based on Theorems~\ref{theo:univote} and \ref{theo:simvote}, then we report the empirical voting output, denoted as Est., and the error between LLM-inferred output and voting output, denoted as Err., i.e., Est. is the ratio of sampled $x_i$ for which $f(t_i)=x_i$, and Err. is the difference between the sampled instances and all instances.
We set $lb=0.15$ and $ub=1-lb$ by default, and set $l=0.9996$, i.e., the failure probability is less than $0.579\%$.

The results are shown in Table~\ref{tab:theoretical}, for both \textsc{UniVote} and \textsc{SimVote}, the practical error is much smaller than the theoretical guaranteed error, and increasing the theoretical error bound $\epsilon$ leads to a substantial reduction in the required sample ratio. Therefore, this supports that small sample ratios are adequate to guarantee strong theoretical error bounds and maintain high empirical effectiveness.
Moreover, although \textsc{SimVote} theoretically yields looser error bounds, its sampling rate $\xi$ is always larger than \textsc{UniVote}. Nevertheless, the empirical performance of \textsc{SimVote} remains comparable to \textsc{UniVote}, suggesting that its semantic voting mechanism is robust even under relaxed theoretical constraints. 

\subsection{Testing on Synthetic Queries}
\label{sec:exp:synthetic}

To evaluate CSV over a broader and more diverse predicate space, we construct a synthetic query suite with controlled variations in semantics and difficulty. 
Concretely, we instruct GPT 5.2 to generate three types of queries for two datasets (IMDB-Review and TC): (1) queries that explicitly refer entities/concepts in the tuple texts (\emph{Explicit} predicates); (2) queries that convey abstract, subjective, interpretive semantics (\emph{Interpretive} predicates); and (3) queries that combine the explicit mentions with interpretive semantics (\emph{Hybrid} predicates). 
For each type, we generate three difficulty levels (easy/moderate/hard), where difficulty reflects semantic subtlety and linguistic complexity.
Each (type, difficulty) group contains 10 distinct queries, yielding 90 queries per dataset and a total of 180 queries across the two datasets. 
We generate these queries by providing sampled tuples and few-shot demonstrations to GPT-5.2. 
We evaluate \CSVU, \CSVS, and linear LLM invocation baseline (\basefilter) on these synthetic queries, respectively. 
For predicates with strong lexical anchoring (Explicit and Hybrid), we use a hybrid clustering distance that combines lexical and semantic signals (weighted sum of BM25 similarity and Euclidean distance),  while for Interpretive predicates, we keep the vanilla Euclidean distance. We use the results of \basefilter as the proxy of ground-truth.

\begin{figure}[t]
    \centering
    \includegraphics[width=0.65\linewidth]{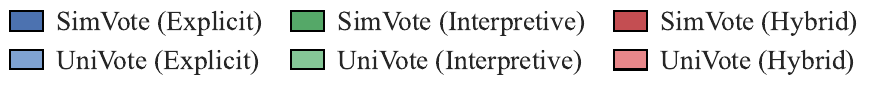}

 \subfigure[IMDB-Review Dataset]{
    \begin{minipage}{\linewidth}
        \includegraphics[width=\linewidth]{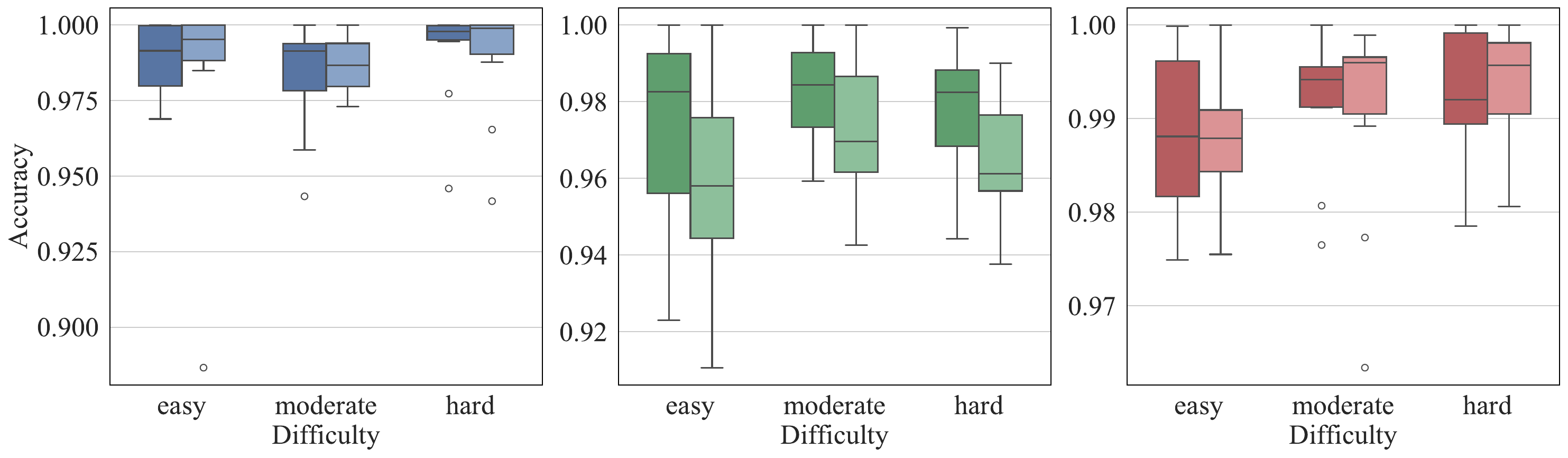}
    \newline
        \includegraphics[width=\linewidth]{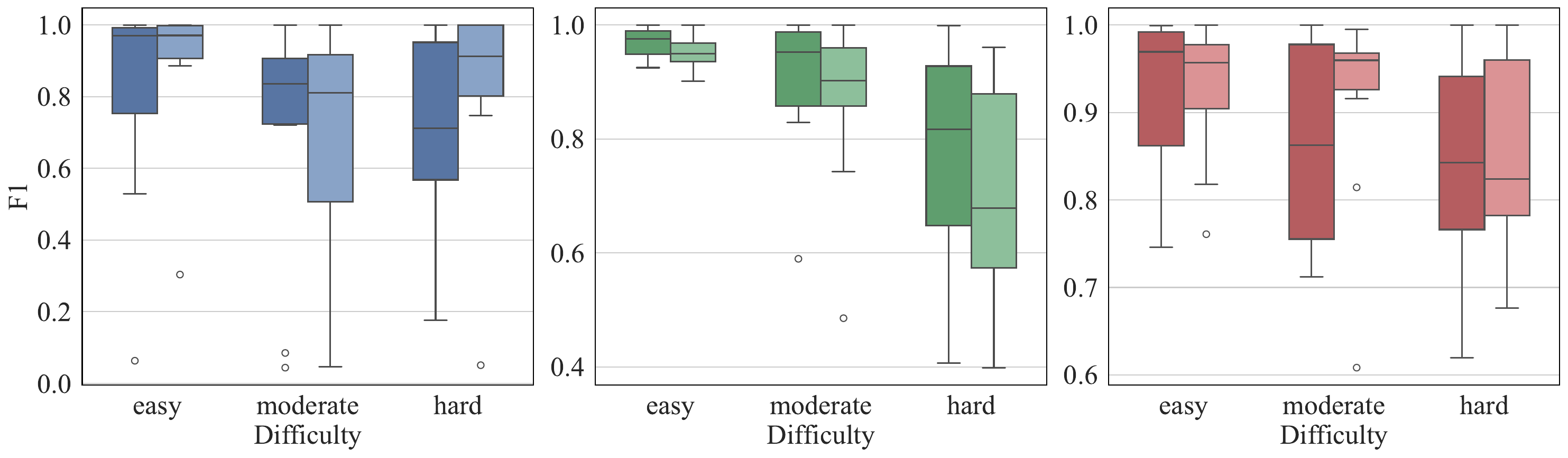}
        \centering
        \includegraphics[width=\linewidth]{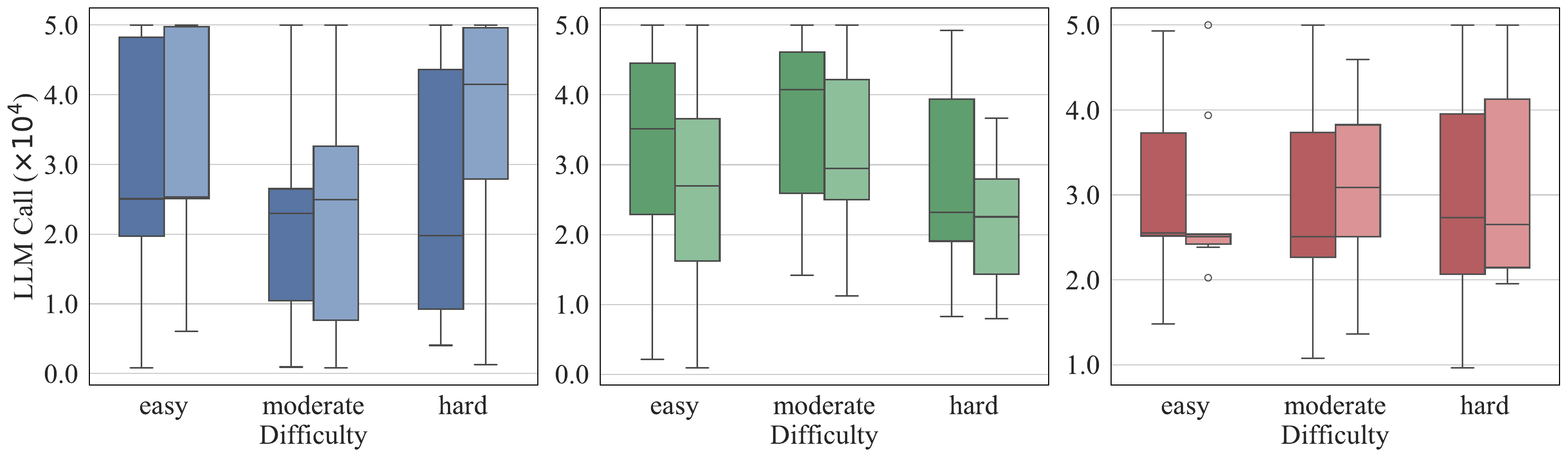}
    \end{minipage}
    }

    \subfigure[Comment Dataset]{
    \begin{minipage}{\linewidth}
        \includegraphics[width=\linewidth]{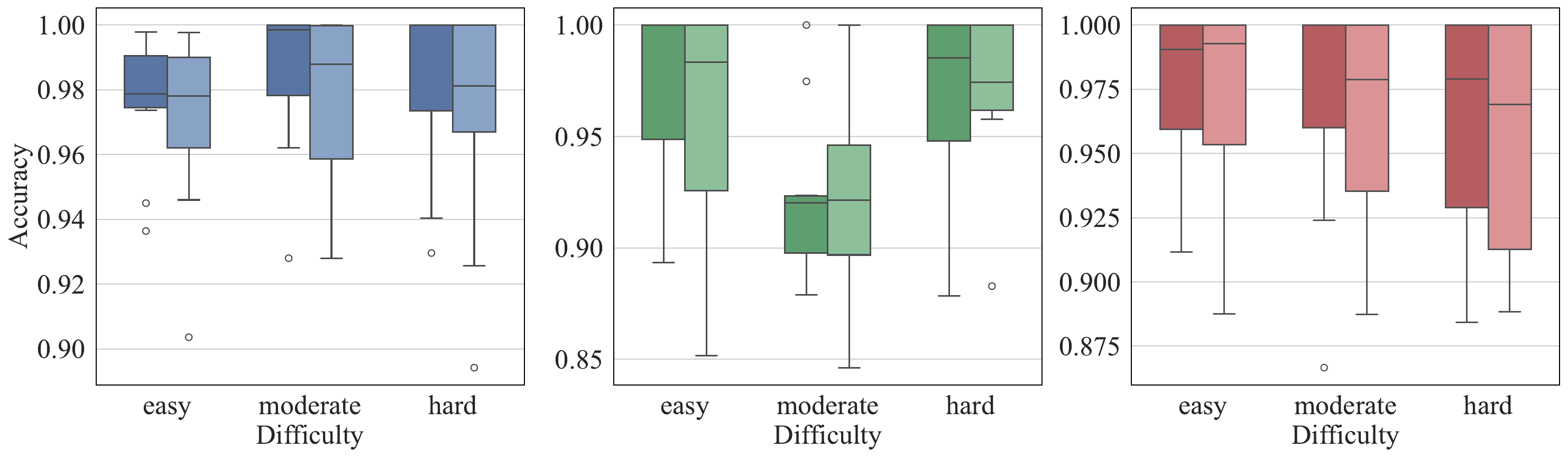}
    \newline
        \includegraphics[width=\linewidth]{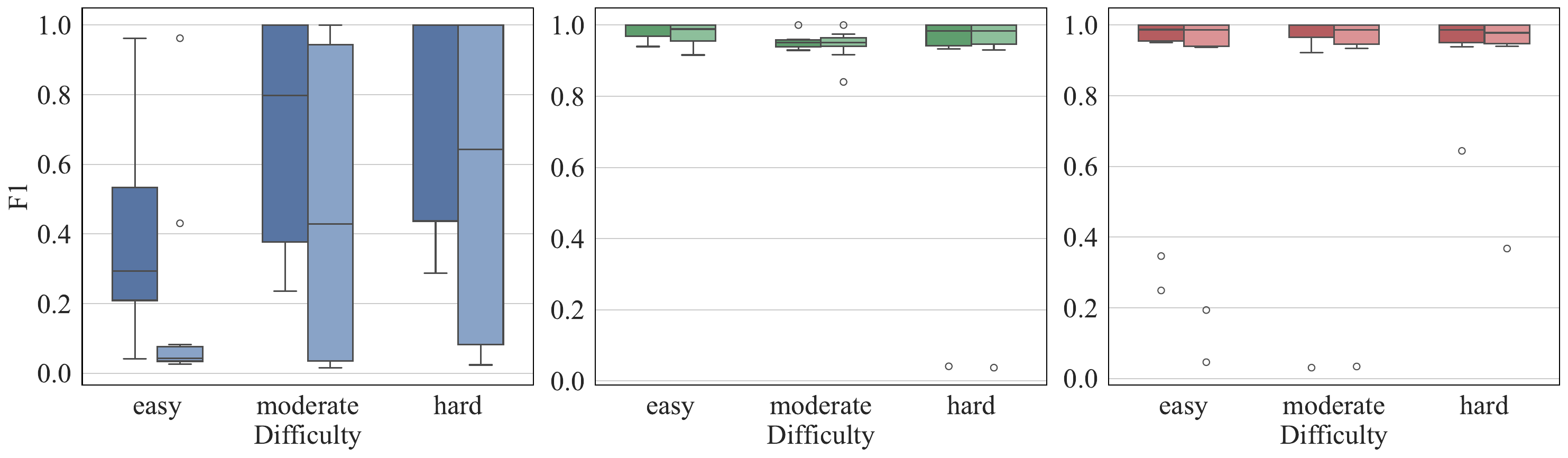}
        \centering
        \includegraphics[width=\linewidth]{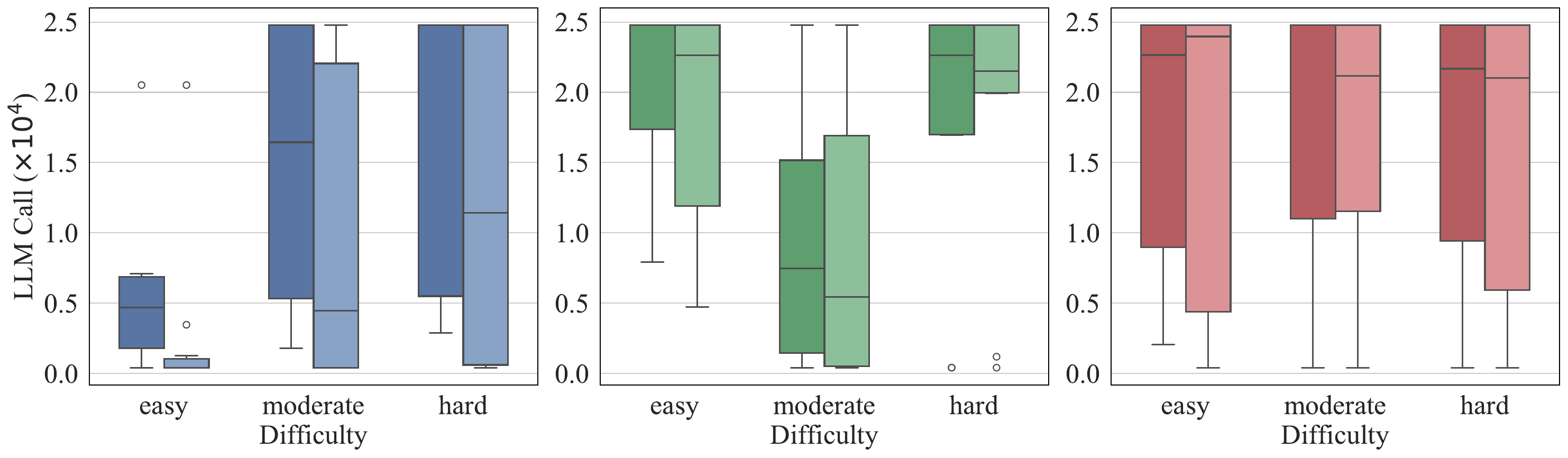}
    \end{minipage}
    }

    \caption{\CSVU and \CSVS on Synthetic Queries}
    \label{fig:auto_test}
\end{figure}

Fig.~\ref{fig:auto_test} reports the box-plot over Accuracy, F1, and \# LLM calls across datasets, query types, and difficulty levels.
Overall, \CSVU and \CSVS achieve high Accuracy across the synthetic suite, while consistently reducing LLM calls relative to the baseline Reference. 
We also observe nuanced failure modes that help characterize the generality of our method. For hard Interpretive/Hybrid predicates, F1 drops relatively to easy/moderate predicates, this is preliminarily caused by extremely low selectivity as there are very few true positives, making the random samples hardly contain enough positives for reliable voting. In such cases, Accuracy can remain high due to label skew (true negatives), while F1 reveals missed rare positives. Similar phenomena also occur in Explicit predicates, where CSV requires more LLM calls to achieve comparable F1 scores. The hybrid of BM25 and embedding distance improves clustering quality for these predicates.
In summary, these results verify that CSV supports a broader predicate space, and reveal a clear limitation regime of our method.

\section{Related Work}
\label{sec:related}

\stitle{General ML-based Analytical Query Processing.}
\sloppy
With the success of deep learning, 
Video database management systems (VDBMS)~\cite{DBLP:journals/pvldb/BangKCMA23, DBLP:journals/pvldb/RomeroHPKZK22, DBLP:journals/pvldb/KangBZ19, DBLP:journals/pvldb/KangGBHZ20, DBLP:journals/pvldb/KangMVBZ20, DBLP:journals/pvldb/XiaoZLWT023} are designed and developed for video data analysis, where deep learning models are used as UDFs in DBMS for image classification, object detection, etc. Oriented for video analytical queries, these systems involve a wide range of optimization techniques, encompassing end-to-end optimization of the preprocessing and query pipeline~\cite{DBLP:journals/pvldb/KangMVBZ20}, declarative query hints~\cite{DBLP:journals/pvldb/RomeroHPKZK22}, fine-grained model selection~\cite{DBLP:conf/sigmod/CaoSHAK22}, query-agnostic semantic index~\cite{DBLP:journals/pvldb/BangKCMA23}, and fast evaluation algorithms for specific operators~\cite{DBLP:journals/pvldb/KangBZ19, DBLP:journals/pvldb/BangKCMA23}.
Apart from video analysis, probabilistic predicates~\cite{DBLP:journals/pvldb/DingAL22, DBLP:conf/sigmod/LuCKC18, DBLP:journals/pvldb/YangWHLLW22} with ML models are used in a broader range of applications, where cheaper proxy models are constructed and their correlations are exploited.
For hybrid queries of ML and relational operators, RAVEN~\cite{DBLP:conf/sigmod/ParkSBSIK22} conducts cross optimization of ML and relational queries, exploiting data statistics to select across multiple executions. 
In addition, some systems also adopt adaptive query processing, which adjusts query plans and resource allocation, or choose cheaper proxy models, following runtime statistics dynamically.
ABAE~\cite{DBLP:journals/pvldb/KangGBHSZ21}, InQuest~\cite{DBLP:journals/pvldb/RussoHK0Z23} focus on optimizing aggregate queries with ML models by introducing sampling-based aggregate algorithms using proxy models.  

\stitle{LLM-based Semantic Query Processing.} LLMs open up a new query paradigm, semantic query processing, for structured data~\cite{DBLP:journals/corr/abs-2504-01157, liu2025optimizing} and unstructured data~\cite{DBLP:conf/icde/WangF25, DBLP:journals/corr/abs-2410-12189, DBLP:journals/corr/abs-2409-00847}. DocETL~\cite{DBLP:journals/corr/abs-2410-12189},  
ZenDB~\cite{DBLP:journals/corr/abs-2405-04674}, DocWanger~\cite{DBLP:journals/corr/abs-2504-14764}, 
Aryn~\cite{DBLP:journals/corr/abs-2409-00847} and Unify~\cite{DBLP:conf/icde/WangF25} leverage LLMs for analytics of unstructured text by semantic queries, where the query languages range from extended SQL dialect~\cite{DBLP:journals/corr/abs-2405-04674}, declarative operations~\cite{DBLP:journals/corr/abs-2410-12189} to natural language~\cite{DBLP:journals/corr/abs-2504-14764, DBLP:conf/icde/WangF25, DBLP:journals/corr/abs-2409-00847}. In these document analytics systems, LLMs are also used as optimization agents in the procedure of query evaluation, such as query planning~\cite{DBLP:journals/corr/abs-2409-00847}, query rewriting, prompt refinement and operation decomposition~\cite{DBLP:journals/corr/abs-2504-14764}. 
UQE~\cite{dai2024uqe} extends SQL by allowing natural language join and selection predicates, which intrinsically support semantic join and semantic filter by LLMs.  
Palimpzest~\cite{liu2025palimpzest} is an LLM-powered analytical system, with logical and physical optimizations to reduce the query time and financial cost of LLM invocations.
Abacus~\cite{russo2025abacus} is an optimizer built on Palimpzest, which integrates a Pareto optimization algorithm and cost estimation for query planning. 
Lotus~\cite{patel2024lotus} formulates LLM-based semantic operations in tables. Furthermore, Lotus proposes optimized algorithms for expensive semantic operations, with batched inference and LLM cascading. 
\BARGAIN \cite{DBLP:journals/pacmmod/ZeighamiSP25} utilizes the same model cascade framework as \lotus. Beyond accuracy alone, it ensures provable bounds on recall, precision, and accuracy for semantic operations.
\cite{akillioglu2025} evaluates the semantic query performance of PostgreSQL with pgAI and DuckDB with FlockMTL~\cite{DBLP:journals/corr/abs-2504-01157}, highlighting the challenges of enforcing structured outputs, batch processing, and query planning.
In these systems, semantic filters are fundamental operators. 

\section{Conclusion}
\label{sec:conclusion}
This paper presents an efficient and scalable approach, Clustering-Sampling-Voting (CSV), to semantic filtering, addressing the core computational challenges of integrating LLMs into data systems. 
By clustering similar tuples, sampling representatives, and inferring labels through theoretical guarantees voting, CSV reduces redundant LLM calls while maintaining comparable effectiveness.
Extensive experiments validate the efficacy of our proposed CSV, showcasing a reduction in the number of LLM calls to 1.28-$200\times$ compared to Reference, and 1.81-$355\times$ compared to Lotus.

\balance
\bibliographystyle{ACM-Reference-Format}
\bibliography{ref}

\newpage
\appendix

\section{Proof of Corollary ~\ref{corollary:1}}
\label{app:proof}

\begin{corollary}
    Let $X = \{x_1, ..., x_n\}$ be a finite population of size $n$, where $x_1, \cdots, x_n$ are binary variables with a mean $\mu$. Let $c_1, ..., c_k$ be sampled without replacement from $X$, and each $c_i$ is associated with a weight $w_i$, where $w_i \ge 0$ and $\sum_{i = 1}^kw_i = 1$. Let $\hat{\mu}$ be the mean of $c_1, \cdots, c_k$, i.e., $\hat{\mu} = \frac{1}{k}\sum_{i = 1}^k w_ic_i$. For any $\epsilon > 0$, we have
    $$
    Pr[| \hat{\mu}_w - \mu| \ge \epsilon] \le 2 \exp(\frac{-3k\epsilon^2}{(6 \hat{\sigma}^2 + 2 \epsilon)v}\cdot \frac{n - k}{n - 1}),
    $$
    where $\hat{\sigma}^2 = \frac{1}{k-1}\sum_{i=1}^k(c_i - \hat{\mu})^2$ is the sample standard deviation, and $v$ is a constant such that $\max_i w_i \le \frac{v}{k}$.
\end{corollary}

\proofsketch
$\hat{\mu}_w - \mu =\sum_{i=1}^k w_ic_i-\sum_{i=1}^k w_i\mu=\sum_{i = 1}^k w_i(c_i-\mu)$. Let $y_i = w_i(c_i - \mu)$, we have $\hat{\mu}_w - \mu=\sum_{i = 1}^k y_i$.
We have the mean of $y_i$ equals to $0$ by the law of big numbers, since $c_i$ is sampled from $X$ with mean $\mu$.
Apply Lemma~\ref{lemma:b-bound} for $y_i$, we have
\begin{align}
    Pr\left[\left|\frac{1}{k}\sum_{i=1}^k y_i\right| \ge \epsilon\right] \le 2\exp(-\frac{k\epsilon^2}{2\hat{\sigma}_y^2 + {2R\epsilon}/3}\cdot\frac{n - k}{n - 1}).\label{eq:cons2}
\end{align}
As $c_i$ i.i.d. of $x_i$, $\hat{\sigma}_y^2=\frac{1}{k-1}\sum_{i=1}^k(w_i(c_i-\mu))^2=\frac{1}{k}\hat{\sigma}^2\sum_{i=1}^kw_i^2$. 
Since $|y_i|=w_i(c_i-\mu)\le 1$, $R \le \max_i w_i \le \frac{v}{k}$, $Pr [| \hat{\mu}_w - \mu| \ge \epsilon]=Pr \left[\left|\sum_{i=1}^k y_i\right| \ge k\epsilon\right]$, and $\sum_{i = 1}^k w_i^2 \le v/k$, According to Eq.~\eqref{eq:cons2}, we have
\begin{align*}
&Pr [| \hat{\mu}_w - \mu| \ge \epsilon]=Pr\left[\left|\sum_{i=1}^k y_i\right|\ge \epsilon\right]=Pr\left[\left|\frac{1}{k}\sum_{i=1}^k y_i\right|\ge \frac{\epsilon}{k}\right]\\
\le&2\exp(-\frac{k(\epsilon/k)^2}{2\hat{\sigma}_y^2 + {2\epsilon}/3k}\cdot\frac{n - k}{n - 1})=2\exp(\frac{-3k\epsilon^2}{({6}\hat{\sigma}^2 + {2\epsilon})v}\cdot\frac{n - k}{n - 1}).
\end{align*}
\eop

\section{Impact of Embedding Models \& LLMs}
\label{app:exp}

\stitle{(1) Impact of Different Embedding Models.}
We study how the embedding model affects \CSVU and \CSVS by comparing BGE-Large-en (BGE)~\cite{bge_embedding}, Qwen-0.6B (Qwen)~\cite{qwen3embedding}, and the default E5-Large~(E5)\cite{wang2022e5} model across all datasets, with results in Table~\ref{tab:model-eff} and Fig.~\ref{fig:model-time}.
For effectiveness, gauged by Accuracy and F1, differences are relatively small. 
BGE is slightly better on RV-Q1 but performs slightly worse on AD-Q2. These results imply that embedding quality has a limited impact on the ultimate result quality.
Efficiency metrics unveil more notable distinctions. BGE leads to increased resource utilization on both RV-Q1 and AD-Q2 due to its embeddings generating less coherent clusters, necessitating more LLM queries for resolution. Qwen performs suboptimally on RV-Q1. Nonetheless, all embedding models maintain significantly enhanced efficiency compared to the baselines, achieving speedups ranging from $1.1\times$ to $18.95\times$, showcasing the robustness of the voting strategy even with suboptimal embedding models.

\stitle{(2) Impact of Different Backbones.}
We evaluate the performance of three backbone LLMs: LlaMA3-8B, a quantized int4 version of LlaMA3-70B, and GPT-4o. The LlaMA models are deployed by local vLLM engine while GPT-4o is used by remote invocation.
The results are detailed in Table~\ref{tab:backbones} and Fig.~\ref{fig:backbone}. All three models exhibit comparable efficiency and effectiveness metrics, suggesting that \CSVU and \CSVS are not overly sensitive to model selection.
LlaMA3-70B excels in performance on RV-Q1, indicating that its larger capacity enables finer semantic discrimination in intricate cases. GPT-4o shows a significant advantage on CB-Q2, surpassing other models by approximately $10\%$ in Accuracy. This improvement likely stems from GPT-4o's superior handling of complex reasoning tasks. However, GPT-4o is not always dominant;  since many queries are adequately resolved by smaller models, and its utilization can sometimes trigger content management limitations.
For efficiency, LLM call volume and token usage remain relatively consistent across CB-Q2 and AD-Q2. On RV-Q1, LlaMA3-70B incurs a higher number of calls, reflecting its sensitivity to nuanced semantic boundaries. Nonetheless, it still achieves a $6.76\times$ speedup over baseline methods. Query latency varies across models due to variations in size, server throughput, and access constraints.

\begin{table}[h]
\centering
\footnotesize
\caption{Semantic Filter with Different LLM Backbones}
    \vspace{-3mm}
\label{tab:backbones}
\resizebox{\linewidth}{!}{
\begin{tabular}{cc cccccc}
\toprule
 \multirow{2}{*}{\bf Algo.} &  \multirow{2}{*}{\bf Model} & \multicolumn{2}{c}{\textbf{RV-Q1}} & \multicolumn{2}{c}{\textbf{CB-Q2}} & \multicolumn{2}{c}{\textbf{AD-Q2}} \\
 & & \textbf{Acc.} & \textbf{F1} & \textbf{Acc.} & \textbf{F1} & \textbf{Acc.} & \textbf{F1} \\
\midrule
\multirow{3}{*}{\CSVU} 
 & LlaMa-8B    & 0.9307  & 0.9336  & 0.7799  & 0.7517  & 0.9190  & 0.8248  \\
 & LlaMa-70B   & 0.9520  & 0.9523  & 0.7715 & 0.8079  & 0.9194  & 0.8251  \\
 & GPT-4o      & 0.9306  & 0.9335 & 0.8870  & 0.8814  & 0.9194   & 0.8250  \\
\midrule
\multirow{3}{*}{\CSVS}
 & LlaMa-8B    & 0.9305   & 0.9335 & 0.7803 & 0.6988  & 0.9193  & 0.8249  \\
 & LlaMa-70B   & 0.9521  & 0.9524 & 0.7701   & 0.8068  & 0.9194   & 0.8251  \\
 & GPT-4o      & 0.9305  & 0.9335 & 0.8666 & 0.8687  & 0.9193  & 0.8250  \\
\bottomrule
\end{tabular}
}
\end{table}

\begin{figure}[h]
    \begin{tabular}[h]{c}
        \centering
        \includegraphics[width=0.7\columnwidth]{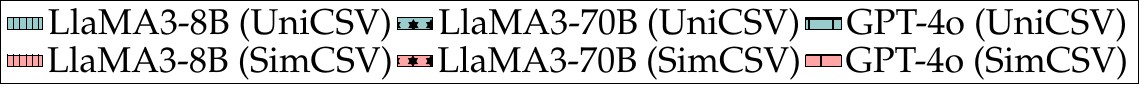}
	\end{tabular}
     \begin{tabular}[h]{c}
        \subfigure[Total \# LLM Calls] {
				\includegraphics[width=0.32\columnwidth]{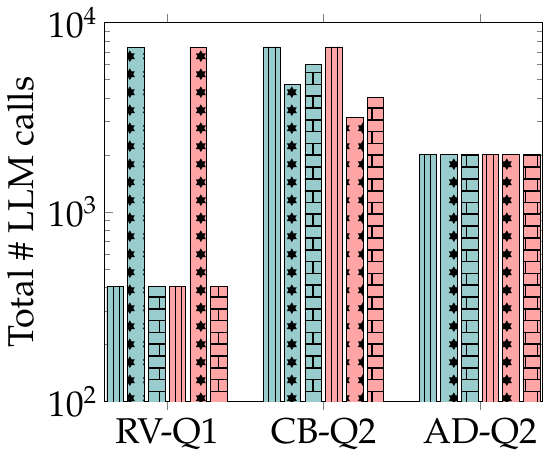}
			\label{fig:T_heatmap}
		} 
        \hspace{-0.3cm}
        \subfigure[Query Time (seconds)] {
			\includegraphics[width=0.32\columnwidth]{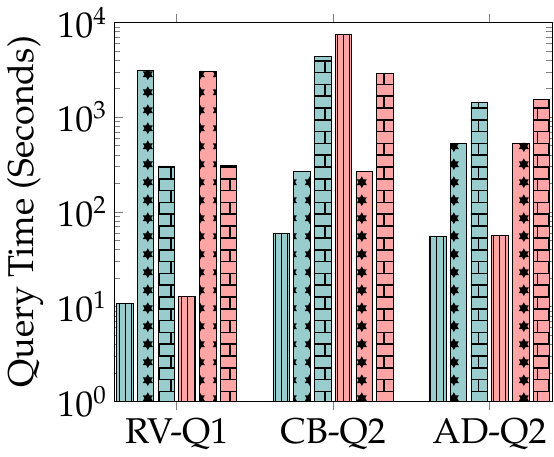}
			\label{fig:F_heatmap}
		}
        \hspace{-0.3cm}
        \subfigure[Total \# Tokens ($\times 10^3$)] {
			\includegraphics[width=0.32\columnwidth]{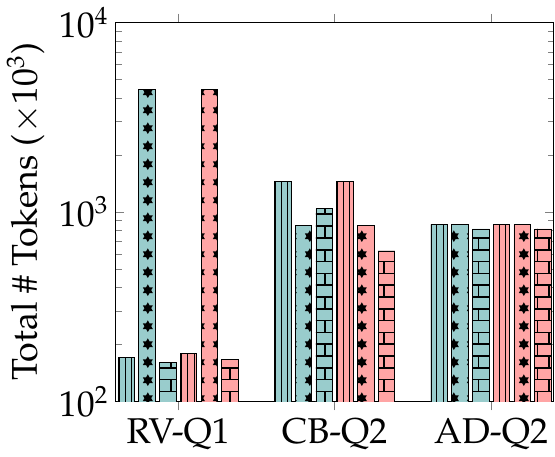}
			\label{fig:F_heatmap}
		}
	\end{tabular}
    \vspace{-5mm}
	\caption{Comparisons on Different LLM Backbones}
        \label{fig:backbone}
\end{figure}

\begin{table}[h]
\centering
\footnotesize
\caption{Semantic Filter with Different Embedding Models}
\label{tab:model-eff}
    \vspace{-3mm}

\begin{tabular}{cr cccccc}
\toprule
 \multirow{2}{*}{\bf Algo.} & \textbf{Model} & \multicolumn{2}{c}{\textbf{RV-Q1}} & \multicolumn{2}{c}{\textbf{CB-Q2}} & \multicolumn{2}{c}{\textbf{AD-Q2}} \\
 & & \textbf{Acc.} & \textbf{F1} & \textbf{Acc.} & \textbf{F1} & \textbf{Acc.} & \textbf{F1} \\
\midrule
\multirow{3}{*}{\CSVU} 
 & E5 & 0.9307 & 0.9336 & 0.7852 & 0.7601 & 0.9192 & 0.8248 \\
 & BGE      & 0.9527 & 0.9531 & 0.7827 & 0.7594 & 0.9122 & 0.8036 \\
 & Qwen  & 0.9227 & 0.9235 & 0.7777 & 0.7540 & 0.9172 & 0.8192 \\
\midrule
\multirow{3}{*}{\CSVS}
 & E5 & 0.9305 & 0.9335 & 0.7815  & 0.7555 & 0.9193 & 0.8249 \\
 & BGE      & 0.9520 & 0.9524 & 0.7806  & 0.7587 & 0.8947 & 0.7755 \\
 & Qwen  & 0.9223 & 0.9231 & 0.7794 & 0.7500 & 0.9172 & 0.8192 \\
\bottomrule
\end{tabular}

\end{table}

\begin{figure}[h]
    \begin{tabular}[t]{c}
        \centering
        \includegraphics[width=1.0\columnwidth]{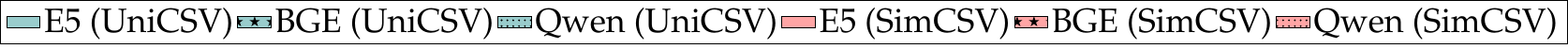}
	\end{tabular}
     \begin{tabular}[t]{c}
        \subfigure[Total \# LLM Calls] {
				\includegraphics[width=0.32\columnwidth]{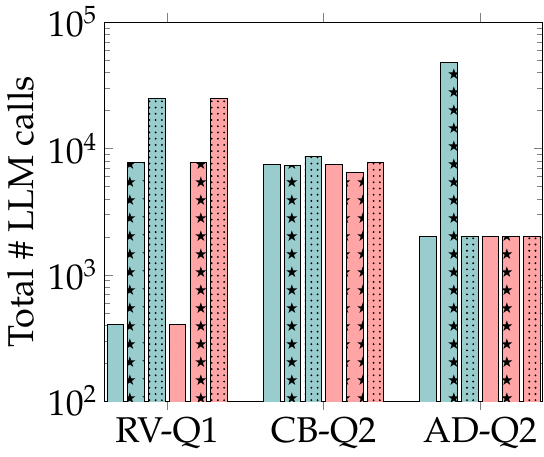}
			\label{fig:T_heatmap}
		} 
        \hspace{-0.3cm}
        \subfigure[Query Time (seconds)] {
			\includegraphics[width=0.32\columnwidth]{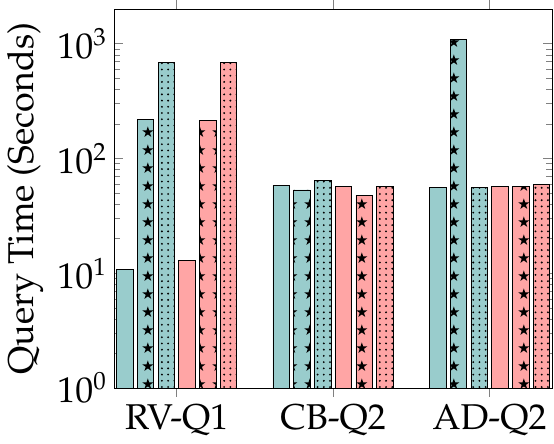}
			\label{fig:F_heatmap}
		}
        \hspace{-0.3cm}
        \subfigure[Total \# Tokens ($\times 10^3$)] {
			\includegraphics[width=0.32\columnwidth]{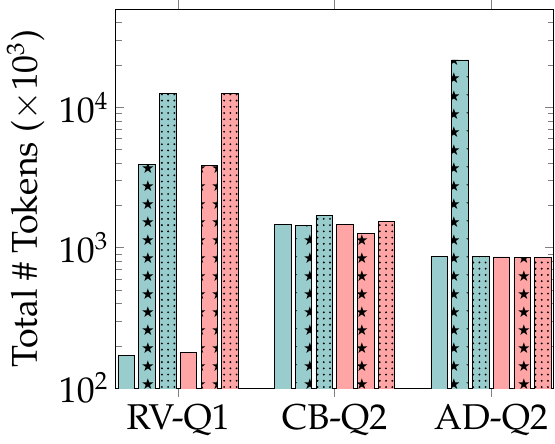}
			\label{fig:F_heatmap}
		}
	\end{tabular}
    \vspace{-5mm}
	\caption{Comparisons of Using Different Embedding Models}
        \label{fig:model-time}
\end{figure}


\end{document}